\documentclass[a4paper]{article}
\usepackage{etoolbox}
\usepackage{fullpage} 

\newtoggle{CONF}
\togglefalse{CONF}

\usepackage{amsmath,amsthm,amssymb}
\usepackage{multirow}

\usepackage{tikz}
\usetikzlibrary{shapes,arrows}
\usetikzlibrary{matrix}

\usepackage{enumerate}
\usepackage{algorithm,algorithmic}

\newcommand{\INITIALLY}{\REQUIRE{}}
\newcommand{\ROUND}{\ENSURE{}}

\newtheoremstyle{newthm}
  {\topsep}   
  {\topsep}   
  {\itshape}  
  {0pt}       
  {\scshape} 
  {.}         
  {5pt}  
  {}          

\theoremstyle{newthm}

\newtheorem{thm}{Theorem}

\newtheorem{lem}[thm]{Lemma}

\DeclareMathOperator{\deaf}{deaf}
\DeclareMathOperator{\dist}{dist}

\DeclareMathOperator{\diam}{diam}

\newcommand{\IR}{\mathbb{R}}

\newcommand{\N}{\mathcal{N}}
\newcommand{\A}{\mathcal{A}}

\newcommand{\E}{_E}
\newcommand{\EP}{_{E'}}

\renewcommand{\leq}{\leqslant}
\renewcommand{\ge}{\geqslant}
\renewcommand{\geq}{\geqslant}

\newcommand{\R}{R}

\renewcommand{\P}{\mathcal{P}}

\usepackage{mathtools}

\DeclareMathOperator\In{In}
\DeclareMathOperator\Out{Out}

\newtheorem{theorem}{Theorem}
\newtheorem{lemma}[thm]{Lemma}
\newtheorem{corollary}[thm]{Corollary}
\theoremstyle{definition}
\newtheorem{definition}[thm]{Definition}

\title{Tight Bounds for Asymptotic and Approximate Consensus}

\author{Matthias F\"ugger$^1$ \and Thomas Nowak$^2$ \and Manfred Schwarz$^3$}

\date{\small
  ${}^1$ CNRS, LSV, ENS Paris-Saclay, Universit\'e Paris-Saclay, Inria\\
  \texttt{mfuegger@lsv.fr}\\
  ${}^2$ Universit\'e Paris-Sud\\
  \texttt{thomas.nowak@lri.fr}\\
  ${}^3$ ECS, TU Wien\\
  \texttt{mschwarz@ecs.tuwien.ac.at}}

\begin{document}
\maketitle

\begin{abstract}
We study the performance of asymptotic and approximate consensus
  algorithms under harsh environmental conditions.
The asymptotic consensus problem requires a set of agents to repeatedly set
  their outputs such that the outputs converge to a common value within the
  convex hull of initial values.
This problem, and the related approximate consensus problem, are fundamental
  building blocks in distributed systems where exact consensus among agents is not
  required or possible, e.g., man-made distributed control systems, and have
  applications in the analysis of natural distributed systems, such as flocking
  and opinion dynamics.
We prove tight lower bounds on the contraction rates of asymptotic consensus algorithms
  in dynamic networks, from which we deduce bounds on the time complexity of approximate
  consensus algorithms.
In particular, the obtained bounds show optimality of asymptotic and approximate
  consensus algorithms presented in [Charron-Bost et al., ICALP'16] for certain
  dynamic networks, including the weakest dynamic network model in which
  asymptotic and approximate consensus are solvable.
As a corollary we also obtain asymptotically tight bounds for asymptotic consensus
  in the classical asynchronous model with crashes.

Central to our lower bound proofs is an extended notion of valency, the set of
  reachable limits of an  asymptotic consensus algorithm starting from a given
  configuration.
We further relate topological properties of valencies to the solvability of
  exact consensus, shedding some light on the relation of
  these three fundamental problems in dynamic networks.
\end{abstract}

\section{Introduction}
In the {\em asymptotic consensus\/} problem a set of agents, each starting
  from an initial value in~$\IR^d$, update their values such that all agents' values
  converge to a common value within the convex hull of initial values.
The problem is closely related to the {\em approximate consensus\/} problem,
  in which agents have to irrevocably decide on values
  that lie within a predefined distance $\varepsilon > 0$ of each other.
The latter is weaker than the {\em exact consensus\/} problem in which 
  agents need to decide on the same value.
Both the asymptotic and the approximate consensus problems have not only a
  variety of applications in the design of man-made control systems
  like sensor fusion~\cite{BS92}, clock synchronization~\cite{LR06}, formation
  control~\cite{EH01},
  rendezvous in space~\cite{LMA05}, or load balancing~\cite{Cyb89}, but also for
  analyzing natural systems like flocking~\cite{VCBCS95}, firefly
  synchronization~\cite{MS90},
  or opinion dynamics~\cite{HK02}.
These problems often have to be solved under harsh environmental restrictions in which
  exact consensus is not achievable, or too costly to achieve:
  with limited computational power and local storage,
  under restricted communication abilities, and in presence of communication uncertainty.

In this work we study the performance of asymptotic and approximate consensus algorithms
  under such harsh conditions.
Specifically, we study algorithms in a {\em network model\/} $\N$
  with round-based computation and a dynamic communication topology whose
  {\em directed\/} communication graphs are chosen each round from a predefined set $\N$ of
  communication graphs.
While this model naturally captures highly unstable communication topologies, we
  later on show that it also allows to assess performance within classical, more stable,
  distributed fault models.

\medskip

\noindent
{\bf Solvability and Algorithms.}
In previous work~\cite{CBFN15}, Charron-Bost et al.\ showed that asymptotic
  consensus is solvable precisely within {\em rooted network models\/}
  in which all communication graphs contain rooted spanning trees.
These rooted spanning trees need not have any edges in common and can change
  from one round to the next.

An interesting special case of rooted network models are network models
  whose graphs are {\em non-split}, that is,
  any two agents have a common incoming neighbor.
Their prominent role is motivated by two properties: (i) They occur as
  communication graphs in benign classical distributed failure models. 
For example, in synchronous systems with crashes, in asynchronous systems with
  a minority of crashes, and synchronous systems with send omissions.
(ii) In~\cite{CBFN15}, Charron-Bost et al.\ showed that non-split graphs also
  play a central role in arbitrary rooted network models: they showed that any
  product of $n-1$ rooted graphs with $n$ nodes is non-split, allowing to
  transform asymptotic consensus algorithms for non-split network models
  into their {\em amortized\/} variants for rooted models.

Interestingly, solvability in any rooted network model is already provided by
  deceptively simple algorithms~\cite{CBFN15}:
  so-called {\em averaging\/} or {\em convex combination\/} algorithms, in which agents repeatedly broadcast
  their current value, and update it to some weighted average of
  the values they received in this round.
One instance, presented by Charron-Bost et al.~\cite{CBFN16} is the
  {\em midpoint algorithm}, in which agents update their value to the midpoint of
  the set of received values, i.e., the average of the smallest and the largest
  of the received values.

Regarding time complexity, for dimension $d=1$, the amortized midpoint algorithm was shown
  to have a contraction rate (of the range of reachable values; see Section~\ref{sec:valency} for a formal definition)
  of $\sqrt[n-1]{{1}/{2}}$ in arbitrary rooted network models with~$n$ agents, and the midpoint
  algorithm of $\frac{1}{2}$ in non-split network models~\cite{CBFN16}.
The latter is \emph{optimal\/} for ``memoryless'' averaging algorithms, which
  only depend on the values received in the current round~\cite{CBFN16}.

A natural question is whether {\em non-averaging or non-memoryless algorithms}, i.e.,
  algorithms that (i) do not necessarily set their output values to within the
  convex hull of previously received values or (ii) whose output is a function not only of the previously received values,
  allow faster contraction rates.
In the context of classical failure models, deriving lower bounds independent of such assumptions, is a long-standing open problem
  raised by Dolev et al.\ in~\cite{DLPSW86}.
As an example for (i), consider the algorithm where each agent sends an equal fraction of its current output value to all out-neighbors
  and sets its output to the sum of values received in the current round.
Note that the algorithm is not a convex combination algorithm as its output may lie outside the convex hull of the values of its in-neighbors.
However, it solves asymptotic consensus algorithm for a fixed directed communication graph.  
Other examples of algorithms that violate (i) and (ii) are from control theory, where the usage of overshooting fast second-order controllers
  is common; see, e.g.,~\cite{Ast10}.
  
\medskip

\noindent
{\bf Contribution.} We prove asymptotically tight lower bounds on the contraction rate of any
  asymptotic consensus algorithm regardless of the structure of the algorithm: algorithms can be full-information and
  agents can set their outputs outside the convex hull of received values.  
This, e.g., includes using higher-order filters in contrast to the $0$-order filters of averaging algorithms.
In particular, the following lower bounds hold for a network model $\N$ with $n$ agents:
If exact consensus is solvable in $\N$, an optimal contraction rate
  of~$0$ can be achieved.
Otherwise:

In a system with $n=2$ agents, the contraction rate is lower bounded by
  $1/3$ (Theorem~\ref{thm:2}).
  This is tight~\cite{CBFN16}.

For an arbitrary communication graph~$G$, we define the set $\deaf(G) = \{F_1,\dots,F_n\}$,
  where $F_i$ is derived from $G$ by making agent $i$ {\em deaf\/}
  in $F_i$, i.e., removing the incoming edges of $i$ in $G$.
  In a system with $n\ge 3$ agents, if $\N$ contains $\deaf(G)$,
  then the contraction rate is lower bounded by $1/2$ (Theorem~\ref{thm:3}).
  This is tight in non-split network models because of the midpoint algorithm~\cite{CBFN16}.

We then show that if $\N$ contains certain rooted graphs $\Psi$,
  the contraction rate is lower bounded by $\sqrt[n-2]{1/2}$ (Theorem~\ref{thm:4}).
  This is asymptotically tight in rooted network models because of the amortized midpoint algorithm~\cite{CBFN16}.
  Specifically, this proves optimality of the amortized midpoint algorithm for the {\em weakest}, i.e., largest,
  network model in which asymptotic and approximate consensus is solvable: the set of all directed rooted communication graphs.
  
For arbitrary network models we show that in a system with $n\ge 3$ agents,
  any asymptotic consensus algorithm must have a contraction rate of at least $1/(D+1)$, where $D$,
  the so-called $\alpha$-diameter of $\N$, i.e.,
  the smallest value which allows a connection of any pair of communication graphs in $\N$ via an indistinguishability chain 
  of length at most $D$ (Theorem~\ref{thm:diam}).
  
We demonstrate how to apply the above mentioned bound to obtain new lower bounds on contraction rates
  for classical failure models as an immediate corollary.
  Specifically, we consider asynchronous message passing system of size $n$ with up to $f < n/2$ crashes.
  For such systems, {\em algorithms operating in asynchronous rounds\/} are widely used~\cite{Lyn96,DLPSW86,CS09}:
    each agent waits for $n-f$ messages corresponding to the current round, updates its state based on the received messages
    and its previous state, and broadcasts the next round's messages.
  
  We show that no algorithm operating in asynchronous rounds can achieve a contraction rate better
    than $\frac{1}{\lceil n/f \rceil+1}$ (Theorem~\ref{thm:crash}).
  This shows that the asynchronous algorithms for systems of size $n > 5f$ with up to $f$ Byzantine failures by Dolev et al.~\cite{DLPSW86}
    and for systems of size $n > 2f$ with up to $f$ crashes by Fekete~\cite{Fek94} have asymptotically optimal contraction rates
    for round-based algorithms.

  We then present an algorithm for $n > f$ that does {\em not operate in asynchronous rounds\/} and achieves a contraction rate of
    $0$, demonstrating a large gap between round-based and non round-based algorithms for asymptotic consensus.
  
Table~\ref{tab:sum} summarizes lower and upper bounds.
Central to the proofs is the concept of the {\em valency of a configuration}
  of an asymptotic consensus algorithm, defined as the set of limits reachable
  from this configuration.
By studying the changes in valency along executions, we infer bounds
  on the contraction rate.

\begin{table*}
\centering
\begin{tabular}{| c | c | c | c | c | c | c|}
\hline
&  \multicolumn{3}{ c| }{network model} &  \multicolumn{2}{ c| }{asynchronous + $f$ crashes}\\
\cline{2-6}
agents & general & non-split with & general & round-based alg.\ & arbitrary alg.\ \\
& non-split & $\alpha$-diameter $D$ & rooted & $0 < f < \frac{n}{2}$  & $0 < f < n$\\
\hline
$n=2$ & $\frac{1}{3}^*$ & $0$ or $\frac{1}{3}^*$ & $\frac{1}{3}^*$ & N/A & \\
\cline{1-5}
$n\ge 3$ & $\frac{1}{2}^*$ & $0$ or $\left[\frac{1}{D+1}^*,\frac{1}{2}\right]$ & $\left[\sqrt[n-2]{\frac{1}{2}}^*,\sqrt[n-1]{\frac{1}{2}}\right]$ & $\left[\frac{1}{\lceil n/f\rceil+1}^*,\frac{1}{\lceil n/f\rceil-1}\right]$ & $0^{*}$\\
\hline
\end{tabular}
\caption{Summary of lower and upper bounds on contraction rates.
  New bounds proved in this work are marked with an ${}^*$. The three left columns are worst-case
  contraction rates for the case the network model is (i) a general non-split,
  (ii) a non-split network model with $\alpha$-diameter $D$, and (iii) a general rooted network model.
  For (ii) contraction rates are $0$ iff exact consensus is solvable.
  The right two columns summarize the bounds for the classical model of an asynchronous system with
    crashes.}
\label{tab:sum}
\end{table*}

We extend the above results on contraction rates to
derive new lower bounds on the decision time of any approximate consensus
algorithm: Let $\Delta > 0$ be the largest distance between initial values.
For $n=2$ we obtain $\lceil\log_3\frac{\Delta}{\varepsilon}\rceil$ (Theorem~\ref{thm:2:approx}).
For $n\ge 3$ and models that include $\deaf(G)$ for a communication graph $G$, we show
  $\lceil\log_2\frac{\Delta}{\varepsilon}\rceil$ (Theorem~\ref{thm:3:approx}),
  and
for $n\ge 4$ and models that include certain $\Psi$ graphs, we obtain
  $(n-2)\lceil\log_2\frac{\Delta}{\varepsilon}\rceil$ (Theorem~\ref{thm:4:approx}).
For arbitrary network models in which exact consensus is not solvable, we show
  $\log_{D+1}\frac{\Delta}{\varepsilon n}$ (Theorem~\ref{thm:diam:approx}).
Again, deciding versions of the asymptotic consensus algorithms
  from~\cite{CBFN16} have matching time complexities;
  showing optimality of these algorithms also for solving approximate consensus. 

\medskip

\noindent
{\bf Related work.} The problem of asymptotic consensus in dynamic networks has been extensively studied in distributed computing
  and control theory,
  see, e.g.,~\cite{EB13,Mor05,Cha11,AB06,CMA08b,BNO08}.
The question of guaranteed convergence rates and decision times of the corresponding approximate consensus problems,
  naturally arise in this context.
Algorithms with convergence times exponential in the number of agents have been
  proposed, e.g., in~\cite{CMA08b}.

Olshevsky and Tsitsiklis~\cite{OT11}, proposed an algorithm with polynomial
  convergence time in bidirectional networks
  with certain stability assumptions on the occurring communication graphs.
The bounds on convergence times were later on refined in~\cite{NOOT09}.
Chazelle~\cite{Cha11} proposed an averaging algorithm with polynomial
  convergence time, which works in any bidirectional connected network model.

To speed up convergence times, algorithms where agents set their output 
  based on values that have been received in rounds prior to the previous round
  have also been considered in literature:
Olshevsky~\cite{Ols15} proposed a linear convergence time algorithm that uses
  messages from two rounds, however, being restricted to
  fixed bidirectional communication graphs.
In~\cite{YSSBG13}, a linear convergence time algorithm for a possibly
  non-bidirectional fixed topology was proposed.
It requires storing all received values.
In previous work~\cite{CBFN16}, Charron-Bost et al.\  presented the midpoint
  algorithm, which has constant convergence time in non-split network models and
  the amortized midpoint algorithm with linear convergence time in rooted
  network models.

To the best of our knowledge, the only lower bound on convergence rate in
  dynamic networks has been shown in~\cite{CSM05}:
  the authors proved that the convergence rate of a specific averaging algorithm
  in a non-split network model with $n$ agents is at least $1-\frac{1}{n}$.

In the context of classical distributed computing failure scenarios, Dolev et al.~\cite{DLPSW86} studied the related approximate consensus problem:
  they considered fully-connected synchronous distributed systems with up to $f$ Byzantine agents, and its asynchronous variant.
The two presented algorithms require $n \ge 3f+1$ for the synchronous and $n \ge 5f+1$ for the asynchronous distributed system,
  the first of which is optimal in terms of resilience~\cite{FLM86}.
The latter result was improved to $n \ge 3f+1$ in~\cite{AAD04}.
Both papers also address the question of optimal contraction rate in such systems.
Since, however, in synchronous systems with $n \ge 3f+1$ exact consensus is solvable, leading to a contraction rate of $0$,
  the authors consider bounds for round-by-round contraction rates.
In~\cite{DLPSW86} they showed that the achieved round-by-round contraction rate of $\frac{1}{2}$ is actually tight for a certain class of algorithms
  that repeatedly set their output to the image of a so-called cautious function applied to the multiset of received values.
A lower bound for arbitrary algorithms, however, remained an open problem.
In higher dimensions, i.e, for any $d\geq 1$, Mendes et al.~\cite{MHVG15}
proposed algorithms with
  convergence time of $d\cdot \lceil \log_2\frac{\sqrt{d}\Delta}{\varepsilon}\rceil$
  under the optimal resiliency condition $n\geq f \cdot (d+2)+ 1$.

Fekete~\cite{Fek90} also studied round-by-round contraction rates for several failure scenarios, again, all in which exact consensus is
    solvable.
He proved asymptotically tight lower bounds for synchronous distributed systems in presence of crashes, omission, and Byzantine agents.
The bounds hold for approximate consensus algorithms that potentially take into account information from all previous rounds.    
In~\cite{Fek94}, Fekete presented an algorithm for asynchronous message-passing systems with minority of crashes, also proving a
  tight lower bound on the contraction rate of any algorithm operating in asynchronous rounds for such systems.
   
\section{Dynamic System Model}\label{sec:model}
 
We consider a set $[n]= \{1,\dots,n \}$ of $n$ agents (also classically called processes).
We assume a distributed, round-based computational model in the spirit
	of the Heard-Of model~\cite{CS09}.
Computation proceeds in {\em rounds}: In every round, each agent sends
	its state to its outgoing neighbors, receives messages from its
        incoming neighbors, and finally updates its state according to
	a deterministic local algorithm, i.e., a transition function that maps the collection
	of incoming messages  to a new state.
Rounds are communication closed in the sense that no agent receives
	messages in round~$t$ that are sent in a round different from~$t$.

Communications that occur in a round are modeled by a {\em directed graph\/} with a node 
	for each agent.
Since an agent can obviously communicate with itself instantaneously, every 
	communication graph contains  a self-loop at each node.	
In the following, we use the {\em product\/} of two communication graphs $G$ and $H$, 
	denoted $G \circ H$, which  is the directed graph  with an edge from $i$ to $j$ 
	 if there exists~$k $ such that $(i,k)$ and $ (k,j) $ are two edges in $G$ and $H$, 
	 respectively.  

We fix a nonempty set of communication graphs $\N$ which determines the 
	{\em network model}.
To fully model dynamic networks in which topology may change continually and 
	unpredictably,  the communication graph at each round  is chosen arbitrarily 
	among~$\N$.
Thus we form the infinite sequences of graphs in~$\N$ which we call 
	{\em communication patterns in\/} $\N$.
In each communication pattern, the communication graph at round~$t$ is denoted 
	by~$G_t$, and $\In_i(t) = \In_i(G_t)$ and $\Out_i(t)=\Out_i(G_t)$ are the sets of incoming 
	and outgoing neighbors (in-neighbors and out-neighbors for short) of
	agent~$i$ in $G_t$.
        
Let us fix an algorithm $\A$; a {\em configuration\/} is a collection of $n$ agent states, one per agent. 
We assume that all agents pick their initial state from the same set of states.
Obviously the picks can be different for different agents.
Since agents are deterministic, given some configuration~$C$ and some communication graph~$G$, 
	the algorithm~$\A$ uniquely determines a new configuration, which  we simply denote~$G.C$ 
	if no confusion can arise.
Then the  {\em execution $E$ of\/} $\A$ from the initial configuration $C_0$ and with the communication
	pattern $\big( G_t \big )_{t \geq 1}$ is the sequence 
	$C_0, G_1,  \dots, C_{t-1}, G_t, C_t, \dots $
	of alternating configurations and communication graphs such that for each round~$t$,
	\mbox{$C_t = G_t. C_{t-1}$}.
The set of executions with communication patterns in $\N$, denoted $\mathcal{E}^{\N}_{\A}$,
	with the distance
	$ \dist(E,E') = 1/2^{\theta}$,  
	where $\theta$ is the first index at which $E$ and $E'$ differ,
	is a compact metric space (e.g., see~\cite{LM95}).
	
Finally, any configuration that occurs in some execution with a communication pattern in $\N$
	is said to be {\em reachable from $C_0$ by $\A$ in\/} $\N$.
In the sequel, the algorithm and the network model are omitted if no confusion can arise.

\subsection{Asymptotic Consensus}

We assume that the local state of agent~$i$ includes a variable $y^i$ in Euclidean 
	$d$-space, and we	
	let $y_{\E}^i (t)\in \IR^d$ denote the value of~$y^i$ at the end of round~$t$
	in execution $E$.
Then we let 
	$ y_{\E}(t) = \big( y_{\E}^1(t), \dots, y_{\E}^n(t) \big)$. 
We write 
\[
\displaystyle\diam(A) = \sup_{x,y\in A} \lVert x-y \rVert
\]
for the diameter of
$A\subseteq \IR^d$ and $\Delta(y(t)) = \diam\{y^1(t),\dots, y^n(t)\}$ for
the diameter of the set of values in round~$t$.

We say an algorithm {\em solves the asymptotic consensus problem\/} in a network model $\N$ 
	if the following holds for every execution $E$ with a communication pattern in $\N$: 
	\begin{itemize}
	\item{\em Convergence.\/} Each sequence $\big(y_{\E}^i (t)  \big)_{t\geq 0}$ converges.
	
	\item{\em Agreement.\/} If $y_{\E}^i(t)$ and $y_{\E}^j(t)$ converge, then they have a common limit.
	
	\item{\em Validity.\/} If $y_{\E}^i(t)$ converges, then its limit is in the convex hull of the
	initial values $y_{\E}^1(0), \dots, y_{\E}^n(0)$.
	
	\end{itemize}
Observe that the {\em consensus function\/} defined by
	$y^* \ : \ E\in (\mathcal{E} , \dist) \mapsto y_{\E}^* \in (\IR^d, \lVert . \lVert)$,
	where $y_{\E}^* $ denotes the common limit of the $n$ sequences $\big( y_{\E}^i(t) \big)_{t\geq 0}$,
	is a priori not continuous.
And indeed, there exist asymptotic consensus algorithms whose consensus
functions are not continuous.

\subsection{Solvability of Asymptotic Consensus with Convex Combination
Algorithms}

In a previous paper~\cite{CBFN15}, Charron-Bost et al.\ proved the following characterization of network models in which 
	asymptotic consensus is solvable.

\begin{thm}[\cite{CBFN15}]\label{thm:CBFN15}
In any dimension $d$, the asymptotic consensus problem is solvable in a network model $\N$ 
	if and only if each graph in $\N$  has a rooted spanning tree.
\end{thm}

For the proof of the sufficient condition, Charron-Bost et al.\ focused on {\em convex combination algorithms\/} where each
   agent $i$ updates its variable $y^i$ to a value within the convex hull of values~$y^j(t-1)$ it has just received.
In particular, they showed in~\cite{CBFN15} that convex combination algorithms where agents update their $y^i$
   via a weighted average of the received values, where weights only depend on the currently received values,
   solve asymptotic consensus in rooted network models.
Such algorithms are memoryless, require little computational overhead and, more importantly, have the benefit of
   working in anonymous networks.
Interestingly, their consensus function $y^*$ is continuous.

\begin{thm}\label{thm:continuity}
The consensus function of every convex combination algorithm
that solves asymptotic consensus is continuous on
the set of its executions.
\end{thm}

\newcommand{\proofthmcontinuity}{
Let $(E_s)_{s\geq 0} $ be a sequence of executions that converges to $E$. 
By definition of the distance on the execution space, this in particular means
that 
\begin{equation}\label{eq:thm:continuity:top}
\forall t\geq 0\   \exists s_t\  \forall s \geq s_t\colon\quad   y_s(0)
= y(0),\ y_s(1) = y(1),\ \dots,\  y_s(t) = y(t) 
\end{equation}
where $y_s(t)$ and $ y(t) $ denote $y_{{\E}_s}(t)$ and $ y_{\E}(t) $,
respectively.

Let $\varepsilon >0$. By definition of the limit~$y^*$ of execution~$E$, there
exists some~$t$ such that
\[
\forall i \in [n]\colon\quad    \lVert y^i (t) - y^*\rVert
\leq
\varepsilon/3 
\enspace.
\]
By~\eqref{eq:thm:continuity:top}, there is an~$s_t$ such that
\[
\forall s\geq s_t\  \forall i \in [n]\colon\quad    \lVert y_s^i (t) - y^*\rVert
\leq
\varepsilon/3 
\enspace.
\]
By the triangle inequality, this means
\[
\forall s\geq s_t\  \forall i,j \in [n]\colon\quad    
\lVert y_s^i (t) - y_s^j(t)\rVert
\leq
2\varepsilon/3 
\enspace.
\]
Because the algorithm is a convex combination algorithm, the limit~$y_s^*$ lies
in the convex hull of the points $y_s^1(t), \dots, y_s^n(t)$.
That is,
\[
\forall s\geq s_t\  \forall i \in [n]\colon\quad    
\lVert y_s^i (t) - y_s^*\rVert
\leq
2\varepsilon/3 
\enspace.
\]
Combining these inequalities gives
\[
\forall s\geq s_t\colon\quad
\lVert y^*_s - y^*\rVert \leq  
\lVert y^*_s - y^i_s(t) \rVert  + 
\lVert y_s^i(t) - y^* \rVert  
\leq \frac{2\varepsilon}{3} + \frac{\varepsilon}{3}
= \varepsilon 
\]
where~$i$ is any agent.
This proves $\displaystyle\lim_{s\to\infty} y^*_s = y^*$ as required.
}
\iftoggle{CONF}{

}{ 
\begin{proof}  
\proofthmcontinuity
\end{proof}
}
  
\section{Valency and Contraction Rate}\label{sec:valency}

We now extend the notion of valency for a consensus algorithm to
	asymptotic consensus algorithms.
We fix an asymptotic consensus algorithm~${\A}$ that solves $d$-dimensional
	asymptotic consensus in a certain network model $\N$ with $n \ge 2$ agents.
Let~$C$ be a configuration reachable by~${\A}$ in $\N$.
Then we define the \emph{valency}  of $C$ by 
  $Y^*_{\N,\A}(C) = \{ y_{\E}^* \in \IR^d \mid  C \mbox{ occurs in  } E \in \mathcal{E}^{\N}_{\A} \}$.

If the algorithm $\A$ is clear from the context, we skip it from the
subscript.
Observe that if $\A$ is a convex combination algorithm, 
	then the valency of a configuration~$C$ is a compact 
	set of $\IR^d$ since the  consensus function is continuous and 
	the set of executions in which $C$ occurs is a compact set.
Set $\delta_{\N}(C) = \diam(Y_{\N}^*(C))$ the diameter of the set
        of reachable limits starting from configuration~$C$.

We have $\delta_{\N}(C_t)\to 0$ in any execution $E = G_0, C_1, G_1, C_2, \dots$
        by Convergence and Agreement.
To study the speed of convergence, we introduce the {\em contraction rate\/} of
algorithm~$\A$ in network model~$\N$ as
\begin{equation*}
\sup_{E\in\mathcal{E}_\A^\N}\limsup_{t \to \infty}\sqrt[t]{\delta_{\N}(C_t)}
\end{equation*}
where $E = C_0, G_1, C_1, G_2, \dots$
In particular, any algorithm that guarantees
  $\delta_{\N}(C_t) \leq \gamma^t \delta_{\N}(C_0)$ for all $t \ge 0$
  has a contraction rate of at most $\gamma$.

We obtain the following for subsets of network models:

\begin{lemma}\label{lem:submodel}
Let $\mathcal{N}, \mathcal{N}'$ be two network models with $\N' \subseteq \N$.
If~$\mathcal{A}$ is an algorithm that solves asymptotic consensus in~$\N$, then 
	(i) it also solves  asymptotic consensus in~$\N'$, (ii) 
	for every configuration~$C$ reachable by~$\mathcal{A}$ in~$\N'$, we have 
	$Y^*_{\N'} (C) \subseteq Y^*_{\N}  (C)$,
        (iii) $\delta_{\N'}(C) \leq \delta_{\N}(C)$, and
        (iv) the contraction rate in $\N'$ is less or equal to the contraction rate in $\N$.
\end{lemma}
\begin{proof}
Statements (i), (ii), and (iii) immediately follow from the definition of
valency.
It remains to show statement (iv).
From $\mathcal{E}_\A^{\N'} \subseteq \mathcal{E}_\A^{\N}$ and~(iii), we deduce
\[
\sup_{E\in\mathcal{E}_\A^{\N'}}\limsup_{t \to \infty}\sqrt[t]{\delta_{\N'}(C_t)}
\leq
\sup_{E\in\mathcal{E}_\A^\N}\limsup_{t \to \infty}\sqrt[t]{\delta_{\N}(C_t)}
\enspace,
\]
which concludes the proof.
\end{proof}
	
We establish two branching properties of valency of configurations in
execution trees.

\begin{lemma}\label{lem:shrink}
Let $C$ be a configuration reachable by algorithm $\A$ in network model $\N$.
Then
\begin{equation*}
Y_{\N}^*(C) = \bigcup_{G\in \N} Y_{\N}^*(G.C)
\enspace.
\end{equation*}
\end{lemma}    
\begin{proof}
First let $y^*\in Y^*_{\N}(C)$.
By definition of $Y^*_{\N}(C)$, there exists an execution $E = C_0, G_1,
C_1, G_2, \dots $ in $\mathcal{E}^{\N}_{\A}$ and a $t\geq 0$ such that
$y^* = y^*_E$ and $C = C_t$.
Set $G=G_{t+1}$.
Hence we have $C_{t+1} = G.C$.
But this shows that $y^* \in Y_\mathcal{N}^*(G.C)$ since $G.C$ occurs in execution~$E$ whose limit is~$y^*$.
This shows inclusion of the left-hand side in the right-hand side.

Now let $G\in \mathcal{N}$ and $y^*\in Y_\mathcal{N}^*(G.C)$.
Then there is an execution $E = C_0, G_1, C_1, G_2, \dots $ in $\mathcal{E}^{\mathcal{N}}_{\mathcal{A}} $
and a $t\geq 0$ such that $y^*= y_{E}^*$ and $G.C = C_t$.
Since $C$ is a reachable configuration, there exists an execution 
$E' = C_0', G_1', C_1', G_2',\dots$ in $\mathcal{E}^{\mathcal{N}}_{\mathcal{A}}$ and an
$s \geq 0$ such that $C_s' = C$.
Then the sequence
$$ E'' = C_0', G_1', \dots, C_s', G, C_t, G_{t+1}, \dots $$
is an execution in $\mathcal{E}^{\mathcal{N}}_{\mathcal{A}} $ with $y^*_{E''} = y^*_E = y^*$.
Hence $y^* \in Y^*_\mathcal{N}(C)$ because~$C$ occurs in~$E''$.
This shows inclusion of the right-hand side in the left-hand side and concludes
the proof.
\end{proof}

\begin{lemma}\label{lem:finiteness}
Let $C$ be a configuration reachable by algorithm $\mathcal{A}$ in network model $\N$.
Then there exist $G,H\in\N$ such that
\begin{equation*}
\diam\big(Y_\mathcal{N}^*(C)\big) = 
\diam\big(Y^*_\N(G.C) \cup Y^*_\N(H.C)\big)
\enspace.
\end{equation*}
\end{lemma}    
\begin{proof}
Set $Y = Y^*_\N(C)$, and $Y_G = Y^*_\N(G.C)$ for $G\in \N$.
By Lemma~\ref{lem:shrink} it is $Y = \bigcup_{G\in\N} Y_G$,
which means that every sequence of pairs of points in~$Y$ whose distances
converge to $\diam(Y)$ includes an infinite subsequence in some product
$Y_G\times Y_H$ because there are only finitely many.
Thus $\diam(Y) \leq \diam(Y_G\cup Y_H)$.
The other inequality follows from $Y_G\cup Y_H\subseteq Y$.
\end{proof}

Two configurations $C$ and $C'$ are called \emph{indistinguishable for agent~$i$},
  denoted $C \sim_i C'$, if $i$ is in the same state in $C$ and in $C'$.
  
As an immediate consequence of the above definition, we obtain:
\begin{lemma}\label{lem:indis}
Let $C$ and $C'$ be two reachable configurations, and let $G$ and $G'$ be communication 
	graphs in $\mathcal{N}$.
If some agent~$i$ has the same in-neighbors in $G$ and $G'$ 
	and if $C \sim_j C'$ for each of $i$'s in-neighbors $j$, then 
	$G.C \sim_i G'.C'$.
\end{lemma}

An agent $i$ is said to be \emph{deaf in a communication graph} $G$ if~$i$ has a unique 
	in-neighbor in~$G$, namely~$i$ itself.
We are now in position to relate valencies of successor configurations.

\begin{lemma}\label{lem:intersect}
If agent~$i$ has the same in-neighbors in two communication graphs $G$ and
$G'$ in $\N$, and if there exists a communication graph in $\N$ in which~$i$
is deaf and $C \sim_{j} C'$ for the in-neighbors $j$ of $i$, then
$Y^*_{\N}(G.C) \cap Y^*_{\N}(G'.C')  \neq \emptyset$ .
\end{lemma}
\begin{proof}
From Lemma~\ref{lem:indis}, we have $G.C \sim_i G'.C'$.

Let $D_i$ be a communication graph in $\N$ in which the agent~$i$ is deaf.
Then we consider  an execution~$E$ in which  $C$ occurs at some round~$t_0-1$, $G$ is the communication
	graph at round~$t_0$, 
	and from there on all communication graphs are equal to $D_i$.
Analogously, let $E'$ be an execution identical to $E$ except that the communication graph at round~$t_0$ 
	is $G'$ instead of $G$. 
By inductive application of Lemma~\ref{lem:indis}, we show that for all $t \geq t_0$, we have
	$C_t \sim_i C'_t $. 
In particular, we  obtain $ y_{\E}^i (t) = y_{\EP}^i(t)$.
Thus $y_{\E}^*= y_{\EP}^*$, which shows that  $Y^*_\mathcal{N}(G.C)$ and
$Y^*_\mathcal{N}(G'.C')$
	intersect.
\end{proof}

From Lemma~\ref{lem:intersect} we determine the valency of any initial configuration
  when the network model contains certain communication graphs.
If every agent is deaf in some communication graph of the network model~$\N$,
  then the next lemma shows that the diameter of the valency of any initial
  configuration is equal to the diameter of the set of its initial values.

\begin{lemma}\label{lem:initial} 
If, for every agent~$i$, there is a communication graph in $\N$ in
  which~$i$ is deaf,
  then each initial configuration $C_0$ satisfies 
$\delta_{\N}(C_0) = \Delta(y(0))$.
In particular, there is an initial configuration for which
$\delta_{\N}(C_0)>0$.
\end{lemma}
\begin{proof}  
Since $Y_{\N}^*(C_0)$ is a subset of the convex hull of
the set of points
$\{y^1(0),\dots,y^n(0)\}$ by the Validity property of asymptotic consensus and
since the diameter of the convex hull of the set $\{y^1(0),\dots,y^n(0)\}$ is equal to
$\Delta(y(0))$, we have the inequality $\delta_\mathcal{N}(C_0) \leq \Delta(y(0))$.

To show the converse inequality, let~$i$ and~$j$ be two agents such that
$\lVert y^i(0) - y^j(0)\rVert = \Delta(y(0))$.
Let $E$ be the execution with initial configuration $C_0$
and a constant communication graph in which agent~$i$ is deaf.
Now consider $C^{(i)}_0$, an initial configuration such that all initial values
are set to $y^i(0)$,
and the execution $E^{(i)}$ from $C^{(i)}_0$ with the same communication pattern as in~$E$.

By a repeated application of Lemma~\ref{lem:indis}, we see that 
at each round~$t$, we have $C_t \sim_i C^{(i)}_t$.
Hence, $ y_{\E}^*= y_{E^{(i)}}^*$.

From the Validity condition, we deduce that  $y^*(E^{(i)}) = y^i(0)$.
It then follows that $y^i(0)  \in Y^*_\mathcal{N} (C_0)$.
By a similar argument, we see
$y^j(0)  \in Y^*_\mathcal{N} (C_0)$.
Hence 
$$\delta_\mathcal{N}(C_0) \geq \lVert y^i(0) - y^j(0)\rVert = \Delta(y(0)) \enspace,$$
which concludes the proof.
\end{proof}

\section{Tight Bound for Two Agents}\label{sec:n2}

In this section, we prove a lower bound of~$1/3$ on the contraction rate
  of algorithms that solve asymptotic consensus in the network model of all rooted (and here also non-split)
  communication graphs with two agents.
Combined with Algorithm~\ref{algo:2:procs}, which achieves this lower bound~\cite{CBFN16}, we
  have indeed identified a tight bound on the contraction rate for $n=2$.
Moreover, the algorithm also shows that the lower bound is achieved by a
  simple convex combination algorithm.

\begin{algorithm}[ht]
\small
\begin{algorithmic}[1]
\REQUIRE{}
  \STATE $y^i \in \IR$
\ENSURE{}
\STATE send $y^i$ to other agent
  \IF{$y^j$ was received from other agent}
    \STATE $y^i \gets y^i/3 + 2y^j/3$
  \ENDIF
\end{algorithmic}
\caption{Algorithm with contraction rate $1/3$ for $n=2$}
\label{algo:2:procs}
\end{algorithm}

A straightforward analysis of Algorithm~\ref{algo:2:procs} shows that its
  contraction rate is equal to~$1/3$.

Note that for $n=2$, there are~$3$ possible rooted communication graphs that
  may occur, all of which are non-split; see Figure~\ref{fig:n2graphs}: (i)~$H_0$ in which all messages are received,
  (ii) $H_1$ in which agent~$2$ receives agent~$1$'s message but not vice versa,
  and (iii) $H_2$ in which agent~$1$ receives agent~$2$'s message but not vice
  versa.

\begin{figure}
\centering
\begin{tikzpicture}[>=latex]
\tikzset{every loop/.style={min distance=5mm,in=0,out=60,looseness=5}}

\node at (-1.2,1.0) {$H_0$};

\node[draw, circle] (n1) at (180:1.1) {$1$};
\node[draw, circle] (n2) at (0:1.1) {$2$};

\draw[->] (n1) edge   [bend left=10] (n2);
\draw[<-] (n1) edge   [bend right=10] (n2);
\path[->] (n1) edge [in=110,out=70,looseness=6] (n1);
\path[->] (n2) edge [in=110,out=70,looseness=6] (n2);

\end{tikzpicture}
\hspace{1.5cm}
\begin{tikzpicture}[>=latex]
\tikzset{every loop/.style={min distance=5mm,in=0,out=60,looseness=5}}

\node at (-1.2,1.0) {$H_1$};

\node[draw, circle] (n1) at (180:1.1) {$1$};
\node[draw, circle] (n2) at (0:1.1) {$2$};

\draw[->] (n1) edge    (n2);
\path[->] (n1) edge [in=110,out=70,looseness=6] (n1);
\path[->] (n2) edge [in=110,out=70,looseness=6] (n2);

\end{tikzpicture}
\hspace{1.5cm}
\begin{tikzpicture}[>=latex]
\tikzset{every loop/.style={min distance=5mm,in=0,out=60,looseness=5}}

\node at (-1.2,1.0) {${H_2}$};

\node[draw, circle] (n1) at (180:1.1) {$1$};
\node[draw, circle] (n2) at (0:1.1) {$2$};

\draw[->] (n2) edge    (n1);
\path[->] (n1) edge [in=110,out=70,looseness=6] (n1);
\path[->] (n2) edge [in=110,out=70,looseness=6] (n2);

\end{tikzpicture}
\caption{The rooted communication graphs $H_0$, $H_1$, and $H_2$ for $n=2$}
\label{fig:n2graphs}
\end{figure}
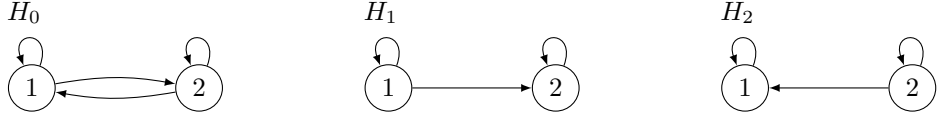

\begin{theorem}\label{thm:2}
The contraction rate of any asymptotic consensus algorithm for $n=2$ agents in
a network model that includes the three graphs $H_0$, $H_1$, and~$H_2$ is
greater or equal to~$1/3$.
\end{theorem}
\begin{proof}  
We show the stronger statement that for every initial configuration~$C_0$ there
is an execution $E=C_0,G_1,C_1,G_2,\dots$ starting from~$C_0$ such that
\begin{equation}\label{eq:thm:2:goal}
\delta_\mathcal{N}(C_t) \geq \frac{1}{3^t} \cdot \delta_\mathcal{N}(C_0)
\end{equation}
for $t\geq 0$.
This, applied to an initial configuration with $\delta_{\N}(C_0)>0$, which exists by
Lemma~\ref{lem:initial}, then shows the theorem. 

Note that it suffices to show~\eqref{eq:thm:2:goal} for the specific network
model $\mathcal{N}' = \{H_0, H_1, H_2\}$ shown in Figure~\ref{fig:n2graphs}
because
$\delta_\mathcal{N}(C_t) \geq \delta_{\mathcal{N'}}(C_t)$
by Lemma~\ref{lem:submodel}
and $\delta_{\mathcal{N}'}(C_0) = \delta_\mathcal{N}(C_0)$ by
Lemma~\ref{lem:initial} whenever $\mathcal{N} \supseteq \mathcal{N}'$.
We hence suppose $\mathcal{N} = \mathcal{N}'$ in the rest of the proof.

The proof is by inductive construction of an execution $E = C_0, G_1,
C_1, G_2, \dots$ whose configurations~$C_t$
satisfy~\eqref{eq:thm:2:goal}.
Equation~\eqref{eq:thm:2:goal} is trivial for $t=0$.

Now assume $t\geq 0$ and that Equation~\eqref{eq:thm:2:goal} holds for
$t$.  There are three possible successor configurations of~$C_{t}$, one
for each of the communication graphs $H_0$, $H_1$, and~$H_2$ in~$\mathcal{N}'$.
Set $C_{t+1}^k = H_k . C_{t}$.
Further let $Y = Y_{\mathcal{N}'}^*(C_{t})$, and $Y_k = Y_{\mathcal{
		N}'}^*(C^k_{t+1})$.

We will show that there is some $\hat{k}\in \{0,1,2\}$ with
$\diam(Y_{\hat{k}}) \ge \diam(Y)/3$.
We then define $G_{t+1} = H_{\hat{k}}$ and $C_{t+1} = C^{\hat{k}}_{t+1}$.
By the induction hypothesis, we then have 
\begin{equation*}
\delta_{\mathcal{N}'}(C_{t+1}) \geq \delta_{\mathcal{N}'}(C_{t})/3 \geq
\delta_{\mathcal{N}'}(C_0)/3^{t+1}
\enspace,
\end{equation*}
i.e., Equation~\eqref{eq:thm:2:goal} holds for~$t+1$.

Assume by contradiction that $\diam(Y_k) < \diam(Y)/3$ for all $k \in
\{0,1,2\}$. 
From Lemma~\ref{lem:shrink} we have $Y=Y_0\cup Y_1\cup Y_2$.
Noting that agent~$1$ is deaf in~$H_1$ and agent~$2$ has the same incoming edges
as in~$H_0$, and that agent~$2$ is deaf in~$H_2$ and agent~$1$ has the same
incoming edges as in~$H_0$,
we obtain from Lemma~\ref{lem:intersect} that
\begin{equation*}
Y_0 \cap Y_1 \neq \emptyset 
\quad
\text{and}
\quad
Y_0 \cap Y_2  \neq \emptyset
\enspace.
\end{equation*}
The sets~$Y_0$ and~$Y_1$ intersecting means
\begin{equation*}
\diam(Y_0\cup Y_1) \leq \diam(Y_0) + \diam(Y_1) < \frac{2}{3} \diam(Y)
\enspace.
\end{equation*}
Further, the sets $Y_0\cup Y_1$ and~$Y_2$ intersecting means
\begin{equation*}
\begin{split}
\diam(Y) = &\diam(Y_0\cup Y_1\cup Y_2)\\ \leq &\diam(Y_0\cup Y_1) + \diam(Y_2) < \diam(Y)
\enspace,
\end{split}
\end{equation*}
a contradiction.
This concludes the proof.
\end{proof}

\section{Tight Bound for Non-split Model: Contraction in Presence of Deaf Graphs}\label{sec:n3}

In this section, we prove a lower bound of~$1/2$ on the contraction rate
   of asymptotic consensus algorithms for $n\ge 3$ agents, in a network model that includes
   graphs derived from a communication graph $G$, where agents are made deaf
   in the derived graphs.
As a special case this includes the network model of all non-split communication graphs.
Charron-Bost et al.~\cite{CBFN16} presented the midpoint algorithm (given in
  Algorithm~\ref{algo:mid}) for dimension one with contraction rate $1/2$ for non-split
  communication graphs.
Together this shows tightness of our lower bound in dimension one. 

\begin{algorithm}
\small
\begin{algorithmic}[1]
\INITIALLY{}
  \STATE $y^i \in \IR$
\ROUND{}
  \STATE send $y^i$ to all agents 
  \STATE $m^i \gets \min\big\{ y^j \mid j\in \In_i(t)\big\}$
  \STATE $M^i \gets \max\big\{ y^j \mid j\in \In_i(t)\big\}$
  \STATE $y^i \gets (m^i+M^i)/2$
\end{algorithmic}
\caption{Midpoint algorithm}
\label{algo:mid}
\end{algorithm}

Let $G$ be an arbitrary communication graph.
Consider a system with $n \geq 3$ agents, and the $n$ communication graphs 
	$F_1,\dots,F_n$ where $F_i$ is obtained by making $i$ deaf in $G$, i.e., by removing all the edges towards~$i$
	except the self-loop $(i,i)$: let $\deaf(G)=\{F_1,\dots,F_n\}$ with
	$F_i = G \setminus\big\{(j , i) \ : \ j \in[n] \setminus \{i\} \big\}$.

With a proof similar to that of Theorem~\ref{thm:2} but noting that
the valencies of all pairs of successor configurations intersect, we get:

\begin{theorem}\label{thm:3}
The contraction rate of any asymptotic consensus algorithm for $n\geq 3$ agents
in a network model that includes $\deaf(G)$
is greater or equal to~$1/2$.
\end{theorem}
\newcommand{\proofthmthree}{
We show the stronger statement that for every initial configuration~$C_0$ there
is an execution $E=C_0,G_1,C_1,G_2,\dots$ starting at~$C_0$ such that
\begin{equation}\label{eq:thm:7:rate}
\delta_{\cal N}(C_t) \ge \frac{1}{2^t}\delta_{\cal N}(C_0) 
\end{equation}
for all $t\geq 0$.
It suffices to show~\eqref{eq:thm:7:rate} for the specific network
model ${\cal N'} = \deaf(G)$ because
$\delta_{\cal N}(C_t) \geq \delta_{{\cal N'}}(C_t)$
by Lemma~\ref{lem:submodel}
and $\delta_{{\cal N'}}(C_0) = \delta_{\cal N}(C_0)$ by
Lemma~\ref{lem:initial} whenever ${\cal N} \supseteq {\cal N'}$.
We hence suppose ${\cal N} = {\cal N'}$ in the rest of the proof.
The proof is by inductive construction of an execution $E = C_0, G_1,
C_1, G_2,\dots$ whose configurations~$C_t$
satisfy~\eqref{eq:thm:7:rate}.
This, applied to an initial configuration with $\delta_{\N}(C_0)>0$, which exists by
Lemma~\ref{lem:initial}, then shows the theorem. 

For $t=0$ the inequality holds trivially.

Now let $t $ be any positive integer and assume that Equation~\eqref{eq:thm:7:rate} holds
	for $t$.
There are $n$ possible successor configurations based on the  
applicable communication graphs $F_1,\dots, F_n$.
We denote  $C^k_{t+1}=F_k.C_{t}$, for any agent  $k$.
Further let $Y = Y_{{\cal N'}}^*(C_{t})$, and $Y_k = Y_{{\cal
		N'}}^*(C^k_{t+1})$.
  
We will show that there exists some agent $\hat{k}\in [n]$ such that
\begin{equation}\label{diam}
\diam(Y_{\hat{k}}) \ge \diam(Y)/2\enspace.
\end{equation}
We then define $G_{t+1} = F_{\hat{k}}$ 
and $C_{t+1} = C^{\hat{k}}_{t+1}$.
By \eqref{diam} and the induction hypothesis, we have 
\begin{equation}
\delta_{{\cal N'}}(C_{t+1}) \geq \frac{\delta_{{\cal N'}}(C_{t})}{2} 
 \geq \frac{1}{2^{t+1}}\delta_{{\cal N'}}(C_{0})
\enspace,
\end{equation}
i.e., Equation~\eqref{eq:thm:7:rate} holds for~$t+1$.

Assume by contradiction that for all~$k\in [n]$ $\diam(Y_k) < \diam(Y)/2$.
Recall that agent $i$ is deaf in $F_i$ and has the same in-neighbors in all the communication
   graphs $F_j$ with $j\neq i$.
Since $n\geq 3$, for any pair of agents $i,j$ we may select an agent $\ell$ different
	from $i$ and $j$ such that $\ell$ has the same in-neighbors in $F_i$ as in $F_j$.
Lemma~\ref{lem:intersect} with the assumption that $F_{\ell}$ is in~$\N$ shows
	that for any pair of agents $i,j $, we have
\begin{align}
  Y_i \cap Y_j & \neq \emptyset\enspace.
\end{align}
By Lemma~\ref{lem:finiteness}, there exist $k,k' \in [n]$ such that 
$\diam(Y_k \cup Y_{k'}) = \diam(Y)$.
In particular, we can choose $i=k$ and $j=k'$, which implies that  
\begin{equation}
\begin{aligned}
	 \diam(Y) = &\diam(Y_k \cup Y_{k'}) \leq \diam(Y_k) + \diam( Y_{k'})\\ <
	 &\diam(Y)
\end{aligned}
\end{equation}
which is a contradiction and concludes the proof.
}
\begin{proof}  
\proofthmthree
\end{proof}

Note that the network model $\deaf(K_n)$, where~$K_n$ is the complete digraph
on~$n$ nodes, is a subset of the network model
that contains all non-split communication graphs.
Hence the lower bound holds and, since Algorithm $2$ is applicable,
a tight bound follows.
In fact, it would even be sufficient to
reduce $\deaf(G)$ to the graphs $F_i,F_j,F_l$ for three agents $i,j,l\in [n]$.

\section{Tight Bound for Rooted Model: Contraction in Presence of $\Psi$ Graphs}

We next prove a lower bound of $\sqrt[n-2]{1/2}$ on the contraction rate of 
  asymptotic consensus algorithms for $n\geq 4$ agents.

\begin{figure}
\centering
\begin{tikzpicture}[>=latex,scale=0.8,rotate=0]
\tikzset{every loop/.style={min distance=5mm,in=0,out=60,looseness=5}}

\node[draw, circle] (n1) at (2,-0) {$i$};
\node[draw, circle] (n2) at (1,-0) {$j$};
\node[draw, circle] (n3) at (3,-0) {$l$};
\node[draw, circle] (n4) at (2,-1) {$4$};
\node[draw, circle] (n5) at (2,-2) {$5$};
\node[draw, circle] (n6) at (2,-3) {$6$};

\path[->] (n1) edge (n4);
\path[->] (n4) edge (n5);
\path[->] (n5) edge (n6);
\path[->] (n6) edge [out=130,in=-90] (n2);
\path[->] (n6) edge [out=50,in=-90] (n3);
\path[->] (n2) edge (n4);
\path[->] (n3) edge (n4);

\end{tikzpicture}
\caption{Rooted communication graph $\Psi_i$ for $n=6$}
\label{fig:rootedgraphs}
\end{figure}
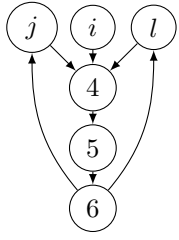

For $i\in \{1,2,3\}$, let $\Psi_i$ (see Figure~\ref{fig:rootedgraphs}) be the communication graph where agents $4 \leq j \leq n-1$ form a path
  with edges from $j$ to $j+1$, agents $\{1,2,3\}\setminus i$ have $n$ as their in-neighbor and $4$ as their out-neighbor,
  and $i$ has $4$ as its out-neighbor.
For $i\in \{1,2,3\}$, let $\sigma_i$ be the sequence of graphs $\Psi_i$ of length $n-2$.

First observe that any communication pattern arising from the concatenation of
$\sigma_i$ sequences necessarily is a communication pattern of the network
model of $\Psi_i$ graphs, which are rooted.
The analysis of the set of these communication patterns necessitates a generalization of our system model:
  generalizing from sets of allowed graphs to arbitrary sets of allowed communication patterns.

\begin{theorem}\label{thm:4}
The contraction rate of any asymptotic consensus algorithm
in a network model including the $\Psi$ graphs
is greater or equal to~$\sqrt[n-2]{1/2}$.
\end{theorem}

From~\cite{CBFN16} we have that the amortized midpoint algorithm guarantees
  a contraction of $\sqrt[n-1]{1/2}$ for rooted network models.
Theorem~\ref{thm:4} shows that this is asymptotically optimal.

\subsection{From Network Models to Sequences}
\label{sec:properties}

To prove Theorem~\ref{thm:4},
we generalize the system model from Section~\ref{sec:model}
and some of the basic lemmas we proved for the specific case
of network models..
While we previously allowed the adversary to choose any sequence of
communication graphs
from the network model, we next consider more general properties on graph
sequences,
including safety and liveness properties.

A {\em property\/} is a set of communication patterns.
A {\em snapshot\/} is a pair $S=({C},\pi)$ where~${C}$ is a
configuration, i.e., a collection of states of the agents, and~$\pi$ is a
finite sequence of communication graphs.
Given a snapshot $S=(C, \pi)$ and a communication graph~$G$,
define $G.S = (G.C, \pi\cdot G)$ where $\pi\cdot G$ is the addition of~$G$ 
to the end of~$\pi$.
We extend this definition to a finite sequence~$\sigma$ of communication graphs. 
We write $S\sim_i S'$ if agent~$i$ has the same local state in both~$S$
and~$S'$.

A {\em trace\/} of an algorithm~$\A$ in a property~$\P$ is an
infinite sequence $T = (S_0, S_1, \dots)$ of snapshots such that
there exists a
communication pattern $P\in\P$ with $S_t = G_t . S_{t-1}$ for
all~$t\geq 1$.
We denote by~$\mathcal{T}_\A^\P$ the set of all traces of~$\A$ in~$\P$.
If~$\A$ solves asymptotic consensus in~$\P$, then we write~$y_T^*$ for the 
common limit of the agents' values in trace $T \in \mathcal{T}_\A^\P$.

We define the valency of snapshots and the contraction rate of an
algorithm in~$\mathcal{P}$ analogously to the case of network models as
\begin{equation*}
Y_{\P}^*(S) = \big\{ y_T^* \in \IR^d \mid S \text{ occurs in } T\in \mathcal{T}_{\A}^{\P} \big\}
\end{equation*}
and the contraction rate as
\begin{equation*}
\sup_{T\in\mathcal{T}_\A^\P}\limsup_{t \to \infty}\sqrt[t]{\delta_{\P}(S_t)}
\end{equation*}
where
$\delta_{\P}(S_t) = \diam\big( Y_\P^*(S_t) \big)$.

\begin{lemma}\label{lem:safety:submodel}
Let $\P, \P'$ be two properties with $\P' \subseteq \P$.
If  $\mathcal{A}$ is an algorithm that solves asymptotic consensus in $\P$, then 
	(i) it also solves  asymptotic consensus in $\P'$, (ii) 
	for every snapshot $S$ reachable by $\mathcal{A}$ in $\P'$, we have 
	$Y^*_{\P'} (S) \subseteq Y^*_{\P}  (S)$,
        (iii) $\delta_{\P'}(S) \leq \delta_{\P}(S)$, and
        (iv) the contraction rate in $\P'$ is less or equal to the contraction rate in $\P$.
\end{lemma}

For a snapshot $S=(C,\pi)$ reachable by algorithm~$\A$ in property~$\P$,
we define $\Sigma(S)$ to be the set of communication graphs~$G$ such that
$\pi\cdot G$ is a prefix of a communication pattern in~$\P$.

\begin{lemma}\label{lem:safety:shrink}
Let $S$ be a snapshot reachable by algorithm $\A$ in property~$\P$.
Then
\begin{equation*}
Y_{\P}^*(S) = \bigcup_{G\in \Sigma(S)} Y_{\P}^*(G.S)
\enspace.
\end{equation*}
\end{lemma}

\begin{lemma}\label{lem:safety:finiteness}
Let $S$ be a configuration reachable by algorithm $\mathcal{A}$ in property~$\P$.
Then there exist $G,H\in\Sigma(S)$ such that
\begin{equation*}
\diam\big(Y_{\P}^*(S)\big) = 
\diam\big(Y^*_{\P}(G.S) \cup Y^*_{\P}(H.S)\big)
\enspace.
\end{equation*}
\end{lemma}

\begin{lemma}\label{lem:safety:intersect}
Let $S = ({C}, \pi)$ and $S'=(C',\pi')$ be two snapshots with $S\sim_i S'$.
If there exist sequences of communication
graphs~$\alpha$ and~$\alpha'$ such that $\pi\cdot\alpha\in\P$,
$\pi'\cdot\alpha' \in\P$ and~$i$ is deaf in all communication graphs
in~$\alpha$ and~$\alpha'$,
then
$Y^*_{\P}(S) \cap Y^*_{\P}(S')  \neq \emptyset$.
\end{lemma}

\begin{lemma}\label{lem:safety:initial} 
Let $\Delta\geq 0$.
If there exist agents $i\neq j$ and communication patterns $P_i,P_j\in\P$ such that
agent~$i$ is deaf in~$P_i$ and agent~$j$ is deaf in~$P_j$,
then there is an initial snapshot $S_0$ with
$\delta_{\P}(S_0) = \Delta$.
In particular, there is an initial snapshot for which
$\delta_{\P}(S_0)>0$.
\end{lemma}

\subsection{Proof of Theorem~\ref{thm:4}}

\begin{lemma}\label{lem:rooted}
For $i,j,\ell \in\{1,2,3\}$ with $\ell \neq i,j:$ $\sigma_i.C_t
\sim_{\ell} \sigma_j.C_t$.
\end{lemma}
\begin{proof}  
We inductively show the following stronger statement.
Let $\sigma^k_i$ be the sequence of graphs $\Psi_i$ of length $k\in[n-2]$.
For agents $i,j,\ell \in \{1,2,3\}$ with $\ell \neq i,j$, and $m \in \{k+3, \dots, n\}$, we have
$\sigma^k_i.C_t \sim_{\ell,m} \sigma^k_j.C_t$.

Observe that agents $\ell$ and $\{4,\dots,n\}$ have the same in-neighbors
  in $\Psi_i$ and $\Psi_j$.
The base case ($k=1$) follows from the observation and Lemma~\ref{lem:indis}.
For the inductive step ($k \mapsto k+1$), observe that agent~$\ell$ 
and $\{k+4,\dots,n\}$ have only incoming edges 
  from agents $\ell$ and $\{k+3,\dots,n\}$. From the hypothesis and Lemma~\ref{lem:indis},
  the inductive step follows.
\end{proof}

Let property $\mathcal{P}_{\mathrm{seq}}$ contain any communication pattern arising
from the concatenation of $\sigma_i$ sequences defined at the start of
the section and property $\mathcal{P}$ contain all communication patterns
generated by rooted graphs.
We show the stronger statement that for every initial snapshot~$S_0$ there
is a trace $T= S_0,S_1,\dots$ starting at~$S_0$ such that
\begin{equation}\label{eq:thm:8:rate}
\delta_\mathcal{P}(S_{t}) \ge
\frac{1}{2^{\lceil\frac{t}{n-2}\rceil}}\delta_\mathcal{P}(S_0) 
\end{equation}
for all $t\geq 0$.
It suffices to show~\eqref{eq:thm:8:rate} for $\P_\mathrm{seq}$ because
$\delta_\mathcal{P}(S_t) \geq \delta_{\P_\mathrm{seq}}(S_t)$
  by Lemma~\ref{lem:safety:submodel}
  and $\delta_{\mathcal{P}_\mathrm{seq}}(S_0) = \delta_\mathcal{P}(S_0)$ by
Lemma~\ref{lem:safety:initial} whenever $\mathcal{P} \supseteq \P_\mathrm{seq}$.
We hence suppose $\mathcal{P} = \P_\mathrm{seq}$ in the rest of the proof.
The proof is by inductive construction of a trace
$T = S_0, S_1, \dots$ whose snapshots~$S_t$
  satisfy~\eqref{eq:thm:8:rate}.
This, applied to an initial snapshot with $\delta_{\P}(S_0)>0$, which exists by
  Lemma~\ref{lem:safety:initial}, then shows the theorem. 

The base case ($t=0$) is trivially fulfilled.

For the inductive step ($t = (n-2)k \mapsto t \leq (n-2)(k+1)$) assume that
  Equation~\eqref{eq:thm:8:rate} holds for $t = (n-2)k$.  
First observe that, by construction of any $P \in \P$, there are
  three possible successor patterns until round $t+n-2$: $\sigma_1, \sigma_2, \sigma_3$.
We thus have,
$Y_{\P}^*(S_{t}) = Y_{\P}^*(S^1_{t+1}) \cup Y_{\P}^*(S^2_{t+1}) \cup Y_{\P}^*(S^3_{t+1})= \dots =
Y_{\P}^*(S^1_{t+n-2}) \cup Y_{\P}^*(S^2_{t+n-2}) \cup Y_{\P}^*(S^3_{t+n-2})$,
where $S^u_{t+n-2}=\sigma_u.S_{t}$ for agent $u\in\{1,2,3\}$.

Abbreviate $Y = Y_{\P}^*(S_{t})$ and $Y_u = Y_{\P}^*(S^u_{t+n-2})$.
We will show that there exists a $\hat{u}\in \{1,2,3\}$ with
\begin{equation}\label{diam:new}
\diam(Y_{\hat{u}}) \ge \diam(Y)/2\enspace.
\end{equation}
We then define $S_{t+n-2} = S^{\hat{u}}_{t+n-2}$.
By \eqref{diam:new} and the induction hypothesis, we then have 
\begin{equation*}
\delta_{\mathcal{P}'}(S_{t+n-2}) = \dots = \delta_{\mathcal{P}'}(S_{t+1}) \geq \frac{\delta_{\mathcal{P}'}(S_{t})}{2} 
 \geq \frac{1}{2^{\lceil\frac{t}{n-2}\rceil+1}}\delta_{\mathcal{P}'}(S_{0})
\enspace,
\end{equation*}
i.e., Equation~\eqref{eq:thm:8:rate} holds up to round~$t+n-2$.

Assume by contradiction that for
  all~$u\in \{1,2,3\}$ $\diam(Y_u) < \diam(Y)/2$.
Since $n\geq 3$ and, by Lemma~\ref{lem:rooted}, $(C_t,\pi.\sigma_i)
\sim_{\ell} (C_t,\pi.\sigma_j)$ together with the fact that
$\pi.\sigma_i.\sigma^{\omega}_{\ell} \in
\P_\mathrm{seq}$ and $\pi.\sigma_j.\sigma^{\omega}_{\ell} \in
\P_\mathrm{seq}$ we can apply Lemma~\ref{lem:safety:intersect} which shows that,
for any pair $i,j\in\{1,2,3\}$
we have
\begin{align*}
  Y_i \cap Y_j & \neq \emptyset\enspace.
\end{align*}
By Lemma~\ref{lem:safety:finiteness}, there exist $u,u' \in \{1,2,3\}$ such that 
$\diam(Y_u \cup Y_{u'}) = \diam(Y)$.
In particular, we can choose $i=u$ and $j=u'$, which implies that  
\begin{align}
  \diam(Y) = &\diam(Y_u \cup Y_{u'}) \leq \diam(Y_u) + \diam( Y_{u'})\notag\\
           < &\diam(Y)\notag
\end{align}
which is a contradiction and concludes the proof.

\section{Relation to Exact Consensus and Generalized Bounds}
\label{sec:rel_exact_consensus}

In~\cite{CGP15}, Coulouma et al.\ characterized the network models in which exact consensus is solvable.
In~\cite{CBFN15}, Charron-Bost et al.\ showed that asymptotic consensus is solvable in a significantly broader class:
  it is solvable if and only if a network model is rooted.
In this section we aim to shed light on the deeper relation between these two problems by studying valencies
   and convergence rates.
Our main results are a characterization of the topological structure of valencies with respect to solvability of exact consensus
  (Theorem~\ref{thm:characterize}) and nontrivial lower bounds on the contraction rates whenever exact consensus is not
  solvable (Theorem~\ref{thm:diam} and Corollary~\ref{cor:diam}).
 
We start with recalling some definitions from Coulouma et al.~\cite{CGP15}.
In the following, we denote by~$\R(G)$ the set of roots of a communication
graph~$G$, i.e., the set of agents that have a directed path to all other
agents in~$G$.
For a set $S \subseteq [n]$, let $\In_S(G) = \bigcup_{j \in S}\In_j(G)$.
The set $\Out_S(G)$ is defined analogously.

\begin{definition}[Definition 4.7 in \cite{CGP15}]
Let~$\N$ be a network model.
Given $G,H,K \in \N$, we define $G \alpha_{\N,K} H$ if
  $\In_{\R(K)}(G) = \In_{\R(K)}(H)$.
The relation $\alpha_\N^*$ is the transitive closure of the union of
  the relations~$\alpha_{\N,K}$ where~$K$ varies in~$\N$.
\end{definition}

\begin{definition}[Definition 4.8 in \cite{CGP15}]
Let~$\N$ be a network model.
We define $\beta_\N$ to be the coarsest equivalence relation included in
$\alpha_\N^*$ such
that for all $G,H$ holds:\\
{\bf(Closure Property)} If $G \beta_\N H$, then there exists
a nonnegative integer~$q$ and communication graphs
$H_0,\dots,H_q\in\N$ and
$K_1, \dots,K_q \in \N$ such that
\begin{itemize}
\item (i) $G= H_0$ and $H=H_q$
\item (ii) $\forall r \in [q] \colon\ H_r \beta_\N G 
	\text{ and }K_r \beta_\N G$
\item (iii) $\forall r \in [q] \colon\ H_{r-1}\alpha_{\N,K_r} H_r$
\end{itemize}
\end{definition}

We next show properties of subsets of network model $\N$ that
  are $\beta_\N$-classes.

\begin{lem}\label{lem:beta:class}
Let~$\N$ be a network model and let $\N'\subseteq \N$ be a
$\beta_\N$-class.
Then 
$G\alpha_{\N'}^* H$ 
and
$G\beta_{\N'} H$ 
for all $G,H\in\N'$.
\end{lem}
\newcommand{\prooflembetaclass}{
Let $G,H\in\N'$.
Since $G\beta_{\N}H$, there is a~$q$ and $H_0,\dots,H_q\in\N$ and 
$K_1,\dots,K_q\in\N$ such that 
(i) $G=H_0$ and $H=H_q$
(ii) $H_r\beta_{\N}G$ and $K_r\beta_{\N}G$ for all $r\in [q]$,
and 
(iii) $H_{r-1}\alpha_{\N,K_r}H_r$ for all $r\in [q]$.
Condition~(ii) implies $H_0,\dots,H_q\in\N'$ and $K_1,\dots,K_q\in\N'$ since
they belong to the same $\beta_{\N}$-class as~$G$, i.e., $\N'$. 
Since all~$H_r$ are in~$\N'$, condition~(iii) can be strengthened to
$H_{r-1}\alpha_{\N',K_r}H_r$ for all $r\in [q]$.

But this means that the pair $(G,H)$ is in the transitive closure of the union
of the relations $\alpha_{\N',K_1}$, \dots, $\alpha_{\N',K_q}$, and thus
in~$\alpha_{\N'}^*$. Hence $\alpha_{\N'}^* = \N'\times \N'$, i.e., the first
part of the lemma.

To show the second part, define relation $\tilde{\beta} = \N'\times \N'$,
which, as we just proved, is included in~$\alpha_{\N'}^*$.
But it also satisfies the closure property in~$\N'$.
Since~$\tilde{\beta}$ is the coarsest equivalence relation on~$\N'$, we thus
have $\beta_{\N'} = \tilde{\beta} = \N'\times \N'$, i.e., the second part of
the lemma.
}
\iftoggle{CONF}{

}{ 
\begin{proof}
\prooflembetaclass
\end{proof}
}

\begin{definition}[Definition 4.5 in \cite{CGP15}]
A network model~$\N$ is called {\em source-incompatible\/} if
\[
\bigcap_{G\in\N} \R(G) = \emptyset
\enspace.
\]
\end{definition}

The proof of Coulouma et al.~\cite{CGP15} actually shows a stronger version
of their theorem (they focus on binary consensus), stated below:

\begin{thm}[Generalization of Theorem 4.10 in \cite{CGP15}]\label{thm:genCG}
Let~$\N$ be a network model.
Exact consensus is solvable in~$\N$ if and only if each
$\beta_{\N}$-class is not source-incompatible.
\end{thm}

We start with showing a generalization of Lemma~\ref{lem:intersect}, that allows us to induce
  non-empty intersection of valencies.

\begin{lem}\label{lem:intersect:alpha}
Let $C$ be a configuration of an asymptotic consensus algorithm $\A$ for $\N$.
For all configurations $C$ in an execution of $\A$ in $\N$,
  and for all $G,H,K\in \N$,
  if $G \alpha_{\N,K} H$ then $Y^*_{\N}(G.C) \cap Y^*_{\N}(H.C) \neq \emptyset$.
\end{lem}
\newcommand{\prooflemintersectalpha}{
By the definition of $G \alpha_{\N,K} H$ it is $\In_{\R(K)}(G) = \In_{\R(K)}(H)$.
Hence, together with Lemma~\ref{lem:indis}, it follows 
that $G.C \sim_i H.C$ for all nodes $i$ in $\R(K)$.
We consider an execution $E$ in which $C$ occurs at some $t_0-1$, $G$ is
the communication graph at $t_0$ and all following graphs are equal to $K$.
Analogously, let $E'$ be an execution identical to $E$ except that the
communication graph at round~$t_0$ is $H$ instead of $G$.  
By inductive application of Lemma~\ref{lem:indis}, we show that for all $t
\geq t_0$, we have $C_t \sim_i C'_t $.  In particular, we  obtain
$ y_{\E}^i(t) = y_{\EP}^i(t) $.
Thus $y_{\E}^*= y_{\EP}^*$, which shows
that  $Y^*_{\cal N}(G.C)$ and $Y^*_{\cal N}(H.C)$ intersect.
}
\iftoggle{CONF}{

}{ 
\begin{proof}
\prooflemintersectalpha
\end{proof}
}

We next establish that for network models in which exact consensus is not solvable,
  asymptotic consensus algorithms must have initial configurations that can be extended
  to executions with different limit outputs.

\begin{lem}\label{lem:initial:noncons}
Let~$\N$ be a network model in which exact consensus is not
solvable.
Then for all asymptotic consensus algorithms~$\A$, there exists an initial
configuration~$C_0$ such that $Y^*_{\N}(C_0)$ is not a singleton.

More precisely, for every~$\Delta>0$, there exists an initial
configuration~$C_0$ such that $\Delta\big(y(0)\big)\leq \Delta$ and
$\delta_{\N}(C_0)\geq \Delta/n$.
\end{lem}
\newcommand{\proofleminitialnoncons}{
We assume without loss of generality that $d=1$.
If not, we embed the initial values in any $1$-dimensional affine subspace.

Let $\N' \subseteq \N$ be any source-incompatible $\beta_\N$-class, which
exists by Theorem~\ref{thm:genCG}.
Consider the $n+1$ initial configurations~$C_0^{(k)}$ where
$0\leq k\leq n$ with initial values
\[
y_i^{(k)}(0) =
\begin{cases}
\Delta & \text{if } i\leq k\\
0 & \text{if } i > k
\enspace.
\end{cases}
\]
For all these initial configurations, we have $\Delta\big(y^{(k)}(0)\big)\leq
\Delta$.
Define $a(k) = \inf Y^*_{\N'}\big(C_0^{(k)}\big)$
and $b(k) = \sup Y^*_{\N'}\big(C_0^{(k)}\big)$.
By Validity, $Y^*_{\N'}\big(C_0^{(0)}\big) = \{0\}$ and 
$Y^*_{\N'}\big(C_0^{(n)}\big) = \{\Delta\}$,
which means $a(0)=b(0) = 0$ and $a(n)=b(n) = \Delta$.
There exists some~$k$ with $1\leq k\leq n$ such that
$b(k-1) \leq b(k) - \Delta/n$ since otherwise $0=b(0) > b(n) - \Delta = 0$.
Because~$\N'$ is source-incompatible, for every agent~$k$, there exists
a communication graph $G^{(k)}\in\N'$ such that $k\not\in S\big(G^{(k)}\big)$.
Since $C_0^{(k-1)} \sim_i C_0^{(k)}$ for all $i\in S\big(G^{(k)}\big)$, 
choosing two executions with all communication graphs equal to~$G^{(k)}$ shows
that $Y^*_{\N'}\big(C_0^{(k-1)}\big) \cap Y^*_{\N'}\big(C_0^{(k)}\big)
\neq \emptyset$, which implies $a(k)\leq b(k-1)$.
Combining both inequalities gives $a(k) \leq b(k) - \Delta/n$ and shows that 
$\delta_{\N'}\big(C_0^{(k)}\big) = b(k)-a(k) \geq \Delta/n$.
We hence choose the initial configuration $C_0 = C_0^{(k)}$.

This shows $\delta_{\N}(C_0) \geq \delta_{\N'}(C_0) \geq \Delta/n$
by Lemma~\ref{lem:submodel} and concludes the proof.
}
\iftoggle{CONF}{

}{ 
\begin{proof}
\proofleminitialnoncons
\end{proof}
}  

This finally allows us to derive one of our main results of this section:
  a characterization of network models in which exact consensus is solvable by
  the topological structure of valencies of asymptotic consensus algorithms.


\begin{theorem}\label{thm:characterize}
Let $\N$ be a network model.
Exact consensus is solvable in $\N$ if and only if
there exists an asymptotic consensus algorithm~$\mathcal{A}$ for~$\N$ 
such that $Y^*_{\N',\mathcal{A}}(C_0)$ is either a singleton or disconnected
for all network models $\N'\subseteq \N$ and all initial configurations~$C_0$
of~$\mathcal{A}$.
\end{theorem}
\newcommand{\proofthmcharacterize}{
($\Rightarrow$):
Assume that exact consensus is solvable in~$\N$, an
let~$\A'$
be an algorithm that solves exact consensus in~$\N$.
Let~$\A$ be the algorithm derived from~$\A'$ in that deciding is
replaced by setting its output variable to the decision value of~$\A'$ and
not changing it anymore.
Before the decision of algorithm~${\cal A}'$, algorithm~$\A$ outputs its
initial value.
Then ${\cal A}$ is an asymptotic consensus algorithm in $\N$.
Further, from Validity of exact consensus, for any initial
configuration $C_0$, the valency $Y^*_{\N,\A}(C_0)$ is a subset of
the set of initial values in $C_0$.
As the set of initial values of~$C_0$ is finite, so is
$Y^*_{\N,\A}(C_0)$ and, by 
Lemma~\ref{lem:submodel}, also $Y^*_{\N',\A}(C_0)$
for all $\N'\subseteq\N$.
Since any finite set is either a singleton or disconnected, the claim follows.

\medskip

($\Leftarrow$):
We assume without loss of generality that $d=1$.
If not, we embed the initial values in any $1$-dimensional affine subspace.

We proceed by means of contradiction.
Assume that exact consensus is unsolvable in~$\N$.
We will show that for all asymptotic consensus algorithms~$\A$ for~$\N$, there
exists an initial configuration~$C_0$
and a network model $\N' \subseteq \N$ such that 
$Y^*_{\N',\A}(C_0)$ is a nontrivial interval.

By Theorem~\ref{thm:genCG}, there is a source-incompatible
$\beta_{\N}$-class.
Choose~$\N'$ to be equal to such a class.
We choose~$C_0$ via Lemma~\ref{lem:initial:noncons}
such that $Y^*_{\N'}(C_0)$ is not a singleton.

To show that $Y^*_{\N'}(C_0)$ is connected, we assume to the contrary that
it is not and derive a contradiction.
The set $Y^*_{\N'}(C_0)$ not being connected
means
the existence of some $z\not\in Y^*_{\N'}(C_0)$ such that
\begin{equation}\label{eq:gap:initial}
\exists
z_1, z_2\in Y^*_{\N'}(C_0)
\colon\ 
z_1 < z < z_2
\enspace.
\end{equation}
We will inductively construct an execution \\$E=C_0,G_1,C_1,G_2,\dots$
such that
\begin{equation}\label{eq:gap:execution}
\exists
z_1, z_2\in Y^*_{\N'}(C_t)
\colon\ 
z_1 < z < z_2
\end{equation}
for all $t\geq 0$.
Setting $m(t) = \inf Y^*_{\N'}(C_t)$ and 
$M(t) = \sup Y^*_{\N'}(C_t)$,
we then have $m(t)\leq z\leq M(t)$ by~\eqref{eq:gap:execution}
and $M(t)-m(t) = \delta_{\N'}(C_t) \to 0$ by Convergence and Agreement.
Hence $\displaystyle\lim_{t\to\infty} m(t) = \lim_{t\to\infty} M(t) = z$, which means
\[
\lim_{t\to\infty}
Y^*_{\N'}(C_t)
=
\bigcap_{t\geq 0}
Y^*_{\N'}(C_t)
=
\{z\}
\enspace,
\]
where the first equality follows from Lemma~\ref{lem:shrink}.
In particular
$z\in Y^*_{\N'}(C_0)$, which gives the desired contradiction.

It thus suffices to construct execution~$E$ 
satisfying~\eqref{eq:gap:execution}.
Assume that~\eqref{eq:gap:execution} holds for a given~$t\geq0$ and
let $z_1^{(t)}, z_2^{(t)}\in Y^*_{\N'}(C_t)$ with $z_1^{(t)} < z <
z_2^{(t)}$.
By Lemma~\ref{lem:shrink}, it follows that there are communication graphs
$G,H\in \N'$ with
$z_1^{(t)}\in Y^*_{\N'}(G.C)$ and $z_2^{(t)} \in Y^*_{\N'}(H.C)$.
By Lemma~\ref{lem:beta:class}, we have $G \alpha_{\N'}^* H$.
Thus there exists a chain $G  = H_0, H_1,\dots, H_q = H \in \N'$ and 
communication graphs $K_1,\dots,K_q\in\N'$ such that
$H_{r-1} \alpha_{\N',K_r} H_r$ for all $r\in [q]$.
From Lemma~\ref{lem:intersect:alpha} we thus know that 
\begin{equation}\label{eq:Hr:inter}
Y^*_{\N'}(H_{r-1}.C) \cap Y^*_{\N'}(H_r.C) \neq \emptyset
\end{equation}
for all $r\in [q]$.
Set $f(r) = \inf Y^*_{\N'}(H_r.C)$
and $g(r) = \sup Y^*_{\N'}(H_r.C)$
for $r\in \{0,\dots,q\}$,
and
\[
\hat{r} = \min\big\{ r\in\{0,\dots,q\} \mid g(r) > z\big\}
\enspace.
\]
Then $f(0) \leq z_1^{(t)} \leq g(0)$ and $f(q) \leq z_2^{(t)} \leq g(q)$.
The quantity~$\hat{r}$ is finite since $g(q)\geq z_2^{(t)} > z$.
We show $f(\hat{r}) < z$ 
by distinguishing two cases:
\begin{enumerate}
\item $\hat{r} = 0$: 
Then $f(\hat{r}) = f(0) \leq z_1^{(t)} < z$.
\item $\hat{r} \geq 1$:
Then, by~\eqref{eq:Hr:inter} and the definition of~$\hat{r}$, we have
$f(\hat{r})  \leq g(\hat{r}-1)  < z$.
\end{enumerate}
In both cases, we showed $f(\hat{r}) < z < g(\hat{r})$.
Choosing $G_{t+1} = H_{\hat{r}}$ and $C_{t+1} = G_{t+1}.C_t$, we hence
proved~\eqref{eq:gap:execution} for $t+1$.
This concludes the proof.
}

\begin{proof}
\proofthmcharacterize
\end{proof}

We next introduce the {\em $\alpha$-diameter\/} of a network model $\N$, which we will then
  (see Theorem~\ref{thm:diam} and Corollary~\ref{cor:diam}) show to be directly linked to a nontrivial lower
  bound on the contraction rate in $\N$ if exact consensus
  is not solvable in $\N$.
Note, that in case exact consensus is solvable in $\N$, the optimal contraction rate
  always is~$0$, obtained by a reduction argument to exact consensus.

\begin{definition}
Let~$\N$ be a network model.
The {\em $\alpha$-diameter\/} of~$\N$ is the smallest~$D\geq1$ such that for
all $G,H\in\N$ there 
exist communication graphs $H_0,\dots,H_q\in\N$ and 
$K_1,\dots,K_q\in\N$ with $q\leq D$ such that 
$G=H_0$, $H=H_q$, and
$H_{r-1}\alpha_{\N,K_r}H_r$ for all $r\in [q]$.
In case it does not exists we set $D = \infty$.
\end{definition}

Observe, that for the network model $\{H_0, H_1, H_2\}$ from Theorem~\ref{thm:2},
  it is $D=2$.
Further, for network model $\deaf(G)$, where $G$ is an arbitrary communication graph $G$,
  we have $D=1$.
The following theorem and corollary thus generalize Theorems~\ref{thm:2} and~\ref{thm:3} to
  arbitrary network models in which exact consensus is not solvable.
  
\begin{theorem}\label{thm:diam}
Let~$\N$ be a network model in which exact consensus is not
solvable.
The contraction rate of any asymptotic consensus algorithm
in~$\N$
is greater or equal to~$1/(D+1)$
where~$D$ is the $\alpha$-diameter of $\N$.
\end{theorem}\newcommand{\proofthmdiam}{
We show the stronger statement that for every initial configuration~$C_0$ there
is an execution $E=C_0,G_1,C_1,G_2,\dots$ starting at~$C_0$ such that
\begin{equation}\label{eq:thm:diam:rate}
\delta_{\N}(C_t) \ge \frac{1}{(D+1)^t}\delta_{\N}(C_0) 
\end{equation}
for all $t\geq 0$.
This, applied to an initial configuration with $\delta_{\N}(C_0)>0$, which
exists by
Lemma~\ref{lem:initial:noncons}, then shows the theorem. 

For the case $D=\infty$, the above statement follows trivially.
We hence suppose $D < \infty$.
The proof is by inductive construction of an execution $E = C_0, G_1,
C_1, G_2,\dots$ whose configurations~$C_t$
satisfy~\eqref{eq:thm:diam:rate}.

For $t=0$ the inequality trivially holds.

Now let $t$ be any nonnegative integer and assume that
Equation~\eqref{eq:thm:diam:rate} holds for $t$.
By Lemma~\ref{lem:finiteness}, 
there exist $G,H\in\N$ such that
$\diam\big( Y_{\N}^*(C_{t}) \big)
=
\diam\big(Y_{\N}^*(G.C_{t}) \cup Y_{\N}^*(H.C_{t}) \big)$.
Because the $\alpha$-diameter of~$\N$ is equal to~$D < \infty$, there
exist communication graphs $H_0,\dots,H_q\in\N$ and 
$K_1,\dots,K_q\in\N$ with $q\leq D$ such that 
$G=H_0$, $H=H_q$, and
$H_{r-1}\alpha_{\N,K_r}H_r$ for all $r\in [q]$.

Define
$Y = Y_{\N}^*(C_t)$ and
$Y_r = Y_{\N}^*(H_r.C_{t})$.
We have $\diam(Y) = \diam(Y_0 \cup Y_q)$ by choice of $G=H_0$ and $H=H_q$.
We show that there exists some $r\in\{0,\dots,q\}$ such that
$\diam(Y_r) \geq \diam(Y)/(q+1)$ and then 
set $G_{t+1} = H_r$
and $C_{t+1} = H_r.C_{t}$.
Then, by the induction hypothesis, we have 
\begin{equation}
\delta_{\N}(C_{t+1}) 
\geq \frac{\delta_{\N}(C_{t})}{q+1} 
\geq \frac{\delta_{\N}(C_{t})}{D+1} 
\geq \frac{1}{(D+1)^{t+1}}\delta_{\N}(C_{0})
\enspace,
\end{equation}
i.e., Equation~\eqref{eq:thm:diam:rate} holds for~$t+1$.

Assume by contradiction that for all 
$r\in \{0,\dots,q\}$ $\diam(Y_r) < \diam(Y)/(q+1)$.
By Lemma~\ref{lem:intersect:alpha}, we have $Y_{r-1} \cap Y_r\neq \emptyset$
for all $r\in [q]$.
Inductively, we prove
\begin{equation}
\diam \big( \bigcup_{s=0}^r Y_s \big) < \frac{r+1}{q+1} \cdot \diam(Y)
\end{equation}
for all $r\in\{0,\dots,q\}$.
In particular for $r=q$, which leads to $\diam(Y) \leq \diam(Y_0\cup Y_q) <
\diam(Y)$,
which is a contradiction and concludes the proof.
}
\begin{proof}
\proofthmdiam
\end{proof}

Direct application of Theorem~\ref{thm:diam} to a network model $\N$ in which
exact consensus is not solvable may yield a trivial bound
of $0$ in case its $\alpha$-diameter is $\infty$.
Indeed, we can, however, use Lemma~\ref{lem:submodel} to derive a strictly
positive bound for any $\N$ in which exact consensus is not solvable:

\begin{corollary}\label{cor:diam}
Let~$\N$ be a network model in which exact consensus is not solvable.
The contraction rate of any asymptotic consensus algorithm
  in~$\N$ is greater or equal to~$1/(D+1)$
  where~$D$ is the smallest $\alpha$-diameter of
  $\N' \subseteq \N$ in which exact consensus is not solvable.
\end{corollary}
\begin{proof}  
Set $\N' \subseteq \N$ equal to the network model with the
  smallest $\alpha$-diameter in which exact consensus is not solvable.
Applying Theorem~\ref{thm:diam} to $\N'$,
  and Lemma~\ref{lem:submodel}~(iv) to $\N'$ and $\N$ yields the corollary.
\end{proof}

\section{Tight Bounds for Asynchronous Systems with Crashes: the Price of Rounds}

In this section we show that Corollary~\ref{cor:diam} provides a tool
  to clearly separate time complexities of algorithms that operate in rounds
  to general algorithms in the classical static fault model of
  asynchronous message passing systems with crashes.  
Our result applies to algorithms without any restriction:
  we do not make assumptions on the nature of the functions used by the agents, and
  agents are not required to be memoryless.

We start with recalling and adapting notation for the classical asynchronous message passing systems.
We consider a distributed system where agents perform receive-compute-broadcast steps.
An agent may crash, i.e., stop making steps.
Crashes can be unclean: the final broadcast message may be received by a proper subset of correct, i.e., non crashed, agents, only.
Since an agent that crashes stops to make steps, we require Convergence, Validity, and Agreement of asymptotic consensus
  to hold only for the set of correct agents.
Analogously, the consensus function $y^*$, and thus the valencies, are restricted to correct agents only.
Further, we apply the standard convention of measuring time in asynchronous systems, by normalizing
  to the longest end-to-end message delay from a broadcast to the respective receive in an execution.

\subsection{Round-based Algorithms}

An algorithm is said to operate in rounds if each agent waits for $n-f$ messages corresponding
  to the current round, updates its state based on the received messages and its previous state,
  and broadcasts the next round's messages.
Indeed algorithms that operate in rounds are widely used in asynchronous systems;
  see, e.g., \cite{Lyn96,DLPSW86,CS09}.

We next show that Corollary~\ref{cor:diam} can be applied to
  obtain new asymptotically tight bounds for round-based algorithms.
Specifically, we prove a lower bound for asynchronous systems of size $n \ge 3$ with up to $f < n/2$ crashes
  whose agents operate in rounds.

Let us construct the following network model:
Denote by~$\mathcal{G}_n$ the set of communication graphs with~$n$ nodes
and let
\begin{equation*}
\N_A
=
\big\{
G \in \mathcal{G}_n
\mid
\forall i\in[n]\colon
\lvert \In_i(G) \rvert \geq n-f
\big\}
\enspace,
\end{equation*}
for some $f < n/2$.

\begin{lemma}\label{lem:async:diam}
The $\alpha$-diameter of~$\N_A$ is at most~$\lceil n/f\rceil$.
\end{lemma}
\begin{proof}
Let $G,H\in\N_A$.
Setting $q=\lceil n/f\rceil$, we choose the communication graphs~$H_r$
and~$K_r$ defined by
\begin{equation*}
\In_i(H_r) = 
\begin{cases}
\In_i(G) & \text{if } 1\leq i \leq rf\\
\In_i(H) & \text{if } rf+1 \leq i \leq n
\end{cases}
\end{equation*}
and
\begin{equation*}
\In_i(K_r) = [n] \setminus \{i \mid (r-1)f+1 \leq i\leq rf  \}
\end{equation*}

Clearly, it is $H_0=G$ and $H_q=H$.
Since we can write $R(K_r) = [n] \setminus \{i \mid (r-1)f+1 \leq i\leq rf  \}$
and $\In_i(H_{r-1}) = \In_i(H_r)$ for all $i\in K_r$, 
we also have
$H_{r-1} \alpha_{\N_A,K_r} H_r$.
Noting $H_r\in \N_A$ and $K_r\in \N_A$,
this concludes the proof.
\end{proof}

From Lemma~\ref{lem:async:diam} and Corollary~\ref{cor:diam} we immediately obtain
  the lower bound:
  
\begin{theorem}\label{thm:crash}
  The contraction rate for any asymptotic consensus algorithm for $n \ge 3$ agents and at most
  $f < n/2$ crashes that operates in rounds is greater or equal to
  $\frac{1}{\lceil n/f \rceil + 1}$.
\end{theorem}

Note that the contraction rate in Theorem~\ref{thm:crash} is with respect to rounds.
However, we can easily construct an execution where a single round requires $1+\varepsilon$
time for arbitrarily small $\varepsilon > 0$: we assign all messages that are delivered
according to the communication graph of the respective round, delay $1$, and all others
  delay $1+\varepsilon$.
Theorem~\ref{thm:crash} thus also holds for a contraction rate with respect to time.

\subsection{General Algorithms}

We next show that there is an algorithm that does not operate in rounds that
  ensures that all agents' outputs are equal by time $f+1$.
This gives a contraction rate of $0$.

The following algorithm {\em MinRelay\/} is inspired by the exact consensus algorithm for
  synchronous systems with crash faults (see, e.g,.~\cite{Lyn96}), and is based on a non-terminating reliable broadcast protocol:
Initially, at time $0$, each agent $i$ sets $S^i$ to the set containing only its initial value, and broadcasts $S^i$.
Whenever an agent~$i$ receives a set $S \neq S^i$, it sets $S^i \gets S^i \cup S$, 
  $y^i \gets \min(S^i)$, and broadcasts $S^i$.

\begin{theorem}\label{thm:async_algo}
The MinRelay algorithm solves asymptotic consensus in asynchronous message passing systems
  with up to $f < n$ crashes.
Specifically, all correct agents' sets $S^i$, and thus $y^i$, are equal by time $f+1$, and the algorithm's
  contraction rate is $0$.
\end{theorem}
\newcommand{\proofthmasyncalgo}{
We first show equality of sets $S^i$ by time $f+1$.
Assume by means of contradiction that there exist two correct agents $i,j$ with
$S^i \neq S^j$ after time $f+1$.
Then there exists an $x$ in $S^i$ that is not in $S^j$.
We distinguish between two cases:

\medskip

\noindent Case i: $x$ was added to $S^i$ at latest by time $f$.
  By the algorithm and the maximum message delay of $1$, $x$ is added to $S^j$ by time $f+1$;
  a contradiction.

\medskip
  
\noindent Case ii: Otherwise, $x$ was added to $S^i$ after time $f$.
Consider the causal chain of messages that lead to adding $x$ at agent $i$.
By the algorithm, its origin must be a message broadcast at time $0$.
Together with the maximum message delay of $1$, the chain must contain at least $f+1$ broadcasts
  during the time $[0,f]$.
At most $f$ of these broadcasts may be ones where an agent crashed and stopped making steps.
Thus there is at least one broadcast among them that happened at an agent that did not crash during the
  broadcast.
By the maximum message delay $1$, node $j$ received this message by time $f+1$, adding $x$ to $S^j$;
  a contradiction to the assumption.
The claim follows.
  
\medskip

Convergence, Agreement, and Validity follow from equality of all correct agents' $S^i$ after time $f+1$,
  the fact all elements in $S^i$ are initial values, and the properties of the function $\min$.
}
\begin{proof}
\proofthmasyncalgo
\end{proof}

\section{Approximate Consensus}\label{sec:approx}

Alternatively to asymptotic consensus, one may also consider the  {\em approximate consensus\/} problem, in which
	convergence is replaced by a decision in a finite number of rounds and where
	agreement should be achieved with an arbitrarily small error tolerance (see, e.g., \cite{Lyn96}). 
Formally, the local state of~$i$ is augmented with a  variable~$d^i$ initialized to~$\bot$.
Agent~$i$ is allowed to set~$d^i$ to some value $v\neq \bot$ only once, in which case
we say that~$i$ {\em decides}~$v$.
In addition to the initial values~$y^i(0)$, agents initially receive the error
tolerance~$\varepsilon$ and an upper bound~$\Delta$ on the maximum distance of
initial values.
An algorithm {\em solves approximate consensus\/} in $\N$ if for all
$\varepsilon >0$ and all~$\Delta$, 
each execution $E$ with a communication pattern in $\N$ with initial diameter at most~$\Delta$  satisfies:
	\begin{itemize}
	\item{\em Termination.\/} Each agent eventually decides.
	
        \item{\em $\varepsilon$-Agreement.\/} If agents $i$ and $j$ decide $v$
            and $v'$, then we have $\lVert v - v' \lVert \leq \varepsilon$.
	
	\item{\em Validity.\/} If agent $i$ decides $v$, then $v$ is in the convex hull of 
	initial values $y_{\E}^1(0), \dots, y_{\E}^n(0)$.
	
	\end{itemize}

Asymptotic consensus and approximate consensus are clearly closely related.
However, the $\varepsilon$-Agreement condition does not preclude the decisions of a given agent,
	as a function of the error tolerance  parameter $\varepsilon$, to diverge, i.e., 
	a priori may lead to unstable decisions with respect to this parameter.

We next extend our lower bounds on the contraction rate of
  asymptotic consensus to lower bounds on the decision time of approximate
  consensus.
In particular, we show optimality of the decision times of the algorithms presented by
  Charron-Bost et al.~\cite{CBFN16}:
For $n=2$, running Algorithm~1 and deciding $y^i$ after $\lceil \log_3 \frac{\Delta}{\varepsilon}\rceil$
  rounds is optimal (Theorem~\ref{thm:2:approx}).
For $n\ge 3$ and the network model of all non-split graphs, running the midpoint
  algorithm and deciding after
  $\lceil \log_2 \frac{\Delta}{\varepsilon}\rceil$ rounds is optimal (Theorem~\ref{thm:3:approx}).
For $n\ge 4$ and the weakest network model of all rooted graphs, running the amortized midpoint
  algorithm and deciding after
  $(n-1)\lceil \log_2 \frac{\Delta}{\varepsilon}\rceil$ rounds is optimal within a multiplicative term of
  at most $\frac{n-1}{n-2}$ (Theorem~\ref{thm:4:approx}).

We start with the case of two agents in Theorem~\ref{thm:2:approx}.
The proof is by reducing asymptotic consensus to approximate consensus,
  arriving at a contradiction with Theorem~\ref{thm:2} for too fast approximate
  consensus algorithms.

\begin{theorem}\label{thm:2:approx}
Let $\Delta>0$ and $\varepsilon>0$.
In a network model of $n=2$ agents that includes the three communication
graphs~$H_0$, $H_1$,
and~$H_2$, all approximate consensus algorithms have an execution with
initial diameter $\Delta(y(0))\leq \Delta$ and decision
time greater or equal to $\log_3 \frac{\Delta}{\varepsilon}$.
\end{theorem}
\newcommand{\proofthmtwoapprox}{
Assume to the contrary that algorithm~$\A$ solves approximate consensus in some
  network model~$\N \supseteq \{H_0,H_1,H_2\}$ that
  decides in $T <  \log_3 \frac{\Delta}{\varepsilon}$ rounds
  for all vectors of initial values~$y(0)$ with $\Delta(y(0))\leq \Delta$
  and some $\varepsilon > 0$.

Choose any~$y(0)$ with $\Delta(y(0))=\Delta$.
Define algorithm~$\tilde{\A}$ by running algorithm~$\A$, updating~$y$ to the
  agents' decision values in round~$T$, and then running
  Algorithm~\ref{algo:2:procs} with the initial values $y^i(T)=d^i$
  from round $T+1$ on.
Because Algorithm~\ref{algo:2:procs} is an asymptotic consensus algorithm and
  the decision values~$y(T)$ of~$\A$ satisfy the Validity condition of approximate
  consensus, algorithm~$\tilde{\A}$ is an asymptotic consensus algorithm.

Let~$C_0$ be an initial configuration of~$\tilde{\A}$ with initial
  values~$y(0)$.
By the proof of Theorem~\ref{thm:2}, namely~\eqref{eq:thm:2:goal}, there is an
  execution $E=C_0,G_1,C_1,G_2,\dots$ starting from~$C_0$ such that 
\begin{equation*}
\delta_{\N}(C_T) \geq \frac{1}{3^T}\cdot \delta_{\N}(C_0)
\enspace.
\end{equation*}
We have $\delta_{\N}(C_0) = \Delta(y(0)) = \Delta$ by Lemma~\ref{lem:initial}
  and $\delta_{\N}(C_T) \leq \Delta(y(T)) \leq \varepsilon$
  by Validity of Algorithm~\ref{algo:2:procs} and
  $\varepsilon$-Agreement of algorithm~$\A$.
But this means $T\geq \log_3 \frac{\Delta}{\varepsilon}$, a contradiction.
}
\begin{proof}
\proofthmtwoapprox
\end{proof}

With a similar proof, 
  we also get the lower bound for approximate
  consensus with $n\geq 3$ agents:

\begin{theorem}\label{thm:3:approx}
Let $\Delta>0$ and $\varepsilon>0$.
In a network model of $n\geq3$ agents that includes the communication
graphs~$\deaf(G)$,
all approximate consensus algorithms have an execution with
initial diameter $\Delta(y(0))\leq \Delta$ and decision
time greater or equal to $\log_2 \frac{\Delta}{\varepsilon}$.
\end{theorem}

Analogously, for network models with rooted $\Psi$ graphs, 
using~\eqref{eq:thm:8:rate},
we obtain:
\begin{theorem}\label{thm:4:approx}
Let $\Delta>0$ and $\varepsilon>0$.
In a network model of $n\geq 4$ agents that includes the $\Psi$ communication
graphs,
all approximate consensus algorithms have an execution with
initial diameter $\Delta(y(0))\leq \Delta$ and decision
time greater or equal to $(n-2)\log_2 \frac{\Delta}{\varepsilon}$.
\end{theorem}

In case the network model does not include any of the above graphs,
  we obtain the following general bound on the termination time:

\begin{theorem}\label{thm:diam:approx}
Let $\Delta>0$ and $\varepsilon>0$.
In a network model in which exact consensus is not solvable,
  all approximate consensus algorithms have an execution with
  initial diameter $\Delta(y(0))\leq \Delta$ and decision
  time greater or equal to $\log_{D+1} \frac{\Delta}{\varepsilon n}$,
  where $D$ is the $\alpha$-diameter of the network model.
\end{theorem}
\newcommand{\proofthmdiamapprox}{
Assume to the contrary that algorithm~$\A$ solves approximate consensus in some
network model~$\N$ in which exact consensus is not solvable
and that decides in $T <  \log_3 \frac{\Delta}{\varepsilon}$
rounds for all vectors of initial values~$y(0)$ with $\Delta(y(0))\leq \Delta$
and some $\varepsilon > 0$.

Define algorithm~$\tilde{\A}$ by repeatedly running algorithm~$\A$, updating~$y$
to the
agents' decision values in round~$kT$, and then restarting~$\A$ in round~$kT+1$
with the decision values from the previous phase.
Then~$\tilde{\A}$ is an asymptotic consensus algorithm.

Let~$C_0$ be an initial configuration of~$\tilde{\A}$ with 
$\Delta\big(y(0)\big) \leq \Delta$ and $\delta_{\N}(C_0) \geq \Delta/n$
By the proof of Theorem~\ref{thm:diam}, 
namely~\eqref{eq:thm:diam:rate}, there is an
execution $E=C_0,G_1,C_1,G_2,\dots$ starting from~$C_0$ such that 
\begin{equation}
	\delta_{\N}(C_T) \geq \frac{1}{(D+1)^T}\cdot \delta_{\N}(C_0)
\enspace.
\end{equation}
It is $\delta_{\N}(C_0) \leq \Delta(y(0)) \leq \Delta/n$
and $\delta_{\N}(C_T) \leq
\Delta(y(T)) \leq \varepsilon$ by $\varepsilon$-Agreement of algorithm~$\A$.
But this means $T\geq \log_{D+1} \frac{\Delta}{\varepsilon n}$, a
contradiction.
}
\begin{proof}
\proofthmdiamapprox
\end{proof}
From Theorem~\ref{thm:diam:approx} and the fact that
$\mathcal{N}'\subseteq \mathcal{N}$ implies $\mathcal{E}'\subseteq \mathcal{E}$ for the
corresponding sets of executions of algorithm $\A$, we get:

\begin{corollary}\label{cor:diam:approx:sub}
Let $\Delta>0$ and $\varepsilon>0$.
In a network model in which exact consensus is not solvable,
  all approximate consensus algorithms have an execution with
  initial diameter $\Delta(y(0))\leq \Delta$ and decision
  time greater or equal to $\log_{D+1} \frac{\Delta}{\varepsilon n}$,
  where $D$ is the smallest $\alpha$-diameter of a
  network model $\N' \subseteq \N$ in which exact consensus is not solvable.
\end{corollary}
\

\section{Conclusions}\label{sec:conc}
We introduced the notion of valency for asymptotic consensus algorithms,
  generalizing the concept of valency from exact consensus algorithms.
Based on the study of valency diameters along executions we proved lower bounds
  on the contraction rates of asymptotic consensus algorithm in arbitrary network models:
In particular, together with previously published averaging algorithms in \cite{CBFN16}, we
  showed tight bounds for the network model containing all non-split graphs, and the
  weakest network model in which asymptotic consensus is solvable, the network model of all
  rooted graphs.
Furthermore we obtained a general lower bound of $1/(D+1)$
  for any network model in which exact consensus is not solvable; here~$D$ denotes
  the newly introduced $\alpha$-diameter of the network model.
Interestingly, this result also immediately provides new tight lower bounds on classical
  static failure models, as exemplified in the case of asynchronous message-passing
  systems with crashes and shows a fundamental discrepancy in performance between
  round-based and general algorithms.
We finally demonstrated how to obtain corresponding results for approximate
  consensus algorithms.

\section*{Acknowledgments}
We would like to thank Bernadette Charron-Bost for the many fruitful
discussions and her valuable input, which greatly helped improve the paper.
The research was partially funded by the 
Austrian Science Fund (FWF)(https://www.fwf.ac.at/)
projects
SIC (P26436)
and
ADynNet (P28182),
and by the
CNRS (http://www.cnrs.fr/)
project
PEPS DEMO.

\bibliographystyle{plain}
\bibliography{agents}













\end{document}